\documentclass[12pt]{article}
\usepackage[margin=1in]{geometry}
\usepackage{amsmath,amssymb,amsthm,mathtools}
\usepackage{booktabs,longtable,array,multirow,tabularx}
\usepackage{graphicx}
\usepackage{float}
\usepackage{microtype}
\usepackage{xcolor}
\usepackage{hyperref}
\usepackage[authoryear,round]{natbib}
\usepackage{enumitem}
\usepackage{tikz}
\usetikzlibrary{arrows.meta,positioning,shapes.geometric,calc}
\usepackage{setspace}
\usepackage{caption}
\usepackage{subcaption}
\usepackage{pdflscape}
\newcounter{algorithm}
\renewcommand{\thealgorithm}{\arabic{algorithm}}
\newenvironment{algorithm}[1][]{%
  \refstepcounter{algorithm}\begin{center}\begin{minipage}{0.90\textwidth}\small%
  \hrule\smallskip%
}{%
  \smallskip\hrule\end{minipage}\end{center}%
}
\newlength{\algindent}
\newenvironment{algorithmic}[1][]{%
  \small
  \setlength{\leftskip}{0pt}%
}{%
  \normalsize
}

\newcommand{\algorithmtitle}[1]{%
  \noindent\textbf{Algorithm \thealgorithm. #1}%
  \par\smallskip\hrule\smallskip
}

\newcommand{\Require}{%
  \par\noindent\textbf{Require:} %
}

\newcommand{\State}{%
  \par\noindent
}

\newcommand{\For}[1]{%
  \par\noindent\textbf{For} #1\textbf{:}\par
  \addtolength{\leftskip}{\algindent}%
}

\newcommand{\EndFor}{%
  \par
  \addtolength{\leftskip}{-\algindent}%
  \noindent\textbf{End for.}\par
}

\newcommand{\If}[1]{%
  \par\noindent\textbf{If} #1\textbf{:}\par
  \addtolength{\leftskip}{\algindent}%
}

\newcommand{\Else}{%
  \par
  \addtolength{\leftskip}{-\algindent}%
  \noindent\textbf{Else:}\par
  \addtolength{\leftskip}{\algindent}%
}

\newcommand{\EndIf}{%
  \par
  \addtolength{\leftskip}{-\algindent}%
  \noindent\textbf{End if.}\par
}

\newcommand{\Repeat}{%
  \par\noindent\textbf{Repeat:}\par
  \addtolength{\leftskip}{\algindent}%
}

\newcommand{\Until}[1]{%
  \par
  \addtolength{\leftskip}{-\algindent}%
  \noindent\textbf{Until} #1.\par
}

\newcommand{\Return}{%
  \par\noindent\textbf{Return:} %
}

\definecolor{uplbmaroon}{HTML}{7B1113}
\definecolor{uplbgreen}{HTML}{014421}
\definecolor{gold}{HTML}{C9A227}
\hypersetup{colorlinks=true,linkcolor=uplbmaroon,citecolor=uplbgreen,urlcolor=blue}

\newtheorem{theorem}{Theorem}
\newtheorem{proposition}{Proposition}
\newtheorem{corollary}{Corollary}

\theoremstyle{definition}
\newtheorem{definition}{Definition}
\newtheorem{remark}{Remark}

\newcommand{\Reff}{\mathcal R_{\mathrm e}}
\newcommand{\dd}{\mathrm d}

\newcommand{\Prb}{\mathbb P}

\title{\textbf{How Simple Models Can Reproduce}\\
\textbf{Multiple Epidemic Waves}\\
Peak Topology, Minimal Dynamic Complexity,\\
and Inference for Time-Varying Transmission}
\author{Jomar Fajardo Rabajante\\
Institute of Mathematical Sciences, University of the Philippines Los Ba\~nos\\
Laguna, Philippines 4031\\
jfrabajante@up.edu.ph}
\date{}

\begin{document}
\maketitle

\begin{abstract}
Multiple epidemic waves are often taken as evidence that simple compartmental models are structurally inadequate. Yet the more informative question is not whether a simple model can reproduce multiple peaks---a possibility already established in the literature---but how much time-varying structure is minimally required to preserve observed epidemic-wave patterns, and what can legitimately be inferred from such a representation. Using the susceptible--infectious--removed (SIR) model as a canonical case, we develop a framework for {peak-preserving dynamic parsimony} that links dynamical systems, inverse problems, statistical inference, and infectious-disease modeling. We establish exact multi-wave representability under time-varying transmission, characterize conditions for epidemic peaks and resurgence, and demonstrate the incidence-only non-identifiability of time-varying transmission and removal rates. Our main methodological contribution is {minimal peak-preserving complexity} (MPPC): the smallest dimension of a positive time-varying transmission function whose posterior-predictive trajectories preserve epidemiologically salient features, including peak count, timing, height, trough depth, and acceptable whole-curve agreement. We estimate MPPC using peak-aware approximate Bayesian computation and evaluate the framework through six reproducible stress tests, a Philippine COVID-19 reported-case proof of concept, and a short-term future-knot projection analysis. Across diverse data-generating mechanisms, low-dimensional time-varying SIR models can preserve the topology of aggregate epidemic waves with surprisingly modest dynamic complexity. However, accurate trajectory compression does not imply recovery of the underlying mechanism, nor does it guarantee reliable long-term counterfactual prediction. The framework therefore separates four questions that are often conflated: representability, minimal dynamic complexity, mechanistic identifiability, and projection validity, providing a general basis for parsimonious modeling of complex epidemic trajectories.
\end{abstract}


\noindent\textbf{Online Shiny app:} \url{https://jfrabajanteuplb.shinyapps.io/MULTIPEAK}

\section{Introduction}
The classical susceptible--infectious--removed (SIR) model remains one of the most useful low-dimensional representations of epidemic dynamics \citep{kermack1927,brauer2012,keeling2008}. Its constant-parameter form is intentionally restrictive: with a closed susceptible population and a fixed removal rate, it produces a single broad epidemic episode rather than repeated resurgences. Real surveillance series, however, frequently display several peaks. These peaks may arise from changes in behavior, interventions, seasonality, pathogen evolution, spatial seeding, heterogeneous contact patterns, school calendars, vaccination, waning immunity, reporting, or combinations of these mechanisms. The broader reproduction-number literature also makes clear that epidemic growth is intrinsically time-dependent and must be interpreted relative to susceptible depletion and the observation process \citep{fraser2007,wallinga2004,cori2013,diekmann2010}.

A common modeling response is to add mechanistic structure until the model can generate the observed waves. That approach is appropriate when the scientific objective is to estimate or intervene on the added mechanisms. It is not the only valid approach when the objective is trajectory reconstruction, short-horizon forecasting, surveillance summarization, or low-dimensional compression. A simple model supplied with a time-dependent effective transmission rate can absorb the net effect of omitted mechanisms. Inverse epidemic modeling has long shown that time-dependent transmission can be reconstructed from epidemic trajectories and incidence data \citep{pollicott2012,mummert2013}. Recent work further demonstrates that important temporal features of reconstructed transmission can remain robust across multiple compartmental structures and incidence formulations \citep{wang2026inverse}. Flexible time-varying transmission has also been implemented with splines, stochastic processes, change points, intervention-linked reproduction numbers, and other bases in epidemic workflows \citep{bouman2024,dehning2020,flaxman2020,giordano2020}, while multiple Gaussian pulses have been used specifically to reproduce multiple COVID-19 waves \citep{setianto2023}. These results establish an important baseline: neither time-varying $\beta(t)$ nor exact inverse reconstruction is, by itself, a sufficient methodological contribution of this paper. We move beyond that baseline to address another fundamental challenge.

The unresolved question is therefore one of dynamic parsimony: not whether a sufficiently flexible forcing can reproduce a complex trajectory, but how much forcing complexity is minimally needed to preserve the features that matter. For multi-wave epidemics, those features are naturally topological and geometrical---the number of waves, their timing and relative heights, the depth of intervening troughs, and their relationship to the full epidemic trajectory. This turns a familiar curve-fitting problem into a smallest-adequate-model problem. It also exposes an inferential question: when does the fitted $\beta(t)$ summarize the net forcing required by the data, and when can it be interpreted as evidence about a biological, behavioral, or policy mechanism?

These questions are especially relevant to likelihood-free and simulator-based inference, for which approximate Bayesian computation (ABC) provides a natural framework. Foundational developments include rejection/regression ABC, likelihood-free sequential Monte Carlo, and ABC-SMC for dynamical systems \citep{beaumont2002,sisson2007,sisson2009correction,toni2009,delmoral2012,sunnaker2013}, with epidemic-specific tutorials and reviews given by \citet{kypraios2017,minter2019,beaumont2019,li2025}. Peak magnitude and peak timing have already been used as outbreak-level summary statistics in ABC calibration, including chikungunya applications \citep{meyer2023}. Accordingly, we do not claim the first use of peak summaries in ABC. Our contribution is to elevate peak topology to an explicit posterior-predictive adequacy criterion and use that criterion to select the minimum dynamic dimension of an effective transmission function. This creates a direct bridge between model reduction, functional complexity, posterior predictive checking, and epidemic-wave interpretation.

The framework is organized around five linked contributions:
\begin{enumerate}[label=(C\arabic*)]
    \item an incidence-form exact reconstruction result, used as a transparent mathematical baseline for multi-wave representability and explicitly connected to the inverse-transmission literature;
    \item a structural incidence-only non-identifiability result for $\beta(t)$ and $\gamma$, clarifying why perfect trajectory reconstruction need not identify epidemiological mechanism;
    \item peak-capacity results for prevalence and incidence, including a resurgence-forcing inequality that quantifies when increasing effective transmission must overcome susceptible depletion;
    \item the principal methodological contribution, {minimal peak-preserving complexity} (MPPC), which converts unrestricted time-varying transmission fitting into a nested smallest-adequate-complexity problem; and
    \item a peak-aware ABC and posterior-predictive workflow that targets peak count, timing, height, trough depth, cumulative timing, and whole-curve agreement.
\end{enumerate}

The same formulation has different interpretations across disciplines. For dynamical-systems and applied-mathematics readers, MPPC is a model-reduction question about the complexity of forcing required by a low-dimensional state representation. For statisticians, it is a posterior-predictive adequacy problem across nested functional dimensions under partial identifiability. For epidemiologists and infectious-disease modelers, it distinguishes reproduction of epidemic waves from recovery of the mechanisms that generated them. For surveillance and public-health applications, it provides an auditable way to ask whether additional compartments are required by the data or whether a simpler model with explicitly interpreted effective forcing is sufficient for the intended task.

The central message is deliberately cautious: a simple SIR model may be an excellent trajectory compressor even when it is a poor mechanistic description. Multi-wave fit alone neither validates homogeneous mixing nor identifies the causes of resurgences; conversely, the presence of multiple peaks alone does not prove that a higher-dimensional compartmental model is necessary.

\section{Simple SIR dynamics with time-varying transmission}
Consider the closed-population SIR system
\begin{align}
\dot S(t) &= -\beta(t)\frac{S(t)I(t)}{N}, \label{eq:sirS}\\
\dot I(t) &= \beta(t)\frac{S(t)I(t)}{N}-\gamma I(t), \label{eq:sirI}\\
\dot R(t) &= \gamma I(t), \label{eq:sirR}
\end{align}
where $N=S+I+R$ is constant, $\gamma>0$, and $\beta(t)\ge 0$. Define latent incidence
\begin{equation}
    c(t)=\beta(t)\frac{S(t)I(t)}{N}=-\dot S(t),
    \label{eq:incidence}
\end{equation}
and the effective reproduction number
\begin{equation}
    \Reff(t)=\frac{\beta(t)}{\gamma}\frac{S(t)}{N}.
    \label{eq:reff}
\end{equation}
Then
\begin{equation}
    \dot I(t)=\gamma I(t)\{\Reff(t)-1\}.
    \label{eq:Igrowth}
\end{equation}
The distinction between prevalence $I(t)$ and incidence $c(t)$ is essential: surveillance peaks usually concern new reported cases, while the standard threshold $\Reff=1$ directly governs the growth or decline of prevalence.

The model has only three boxes: people who can still be infected, people who are currently infectious, and people who have left the infectious state. A changing $\beta(t)$ is not interpreted automatically as one single biological cause. It is a compact dial summarizing whatever made transmission easier or harder at each time: behavior, policy, seasonality, variants, testing, reporting, mobility, or unmodeled heterogeneity.

Table~\ref{tab:notation} presents the notation used throughout the mathematical development. Rates are per unit time; in the Shiny app implementation the default unit is one day.

\footnotesize
\begin{longtable}{@{}>{\raggedright\arraybackslash}p{0.34\textwidth}>{\raggedright\arraybackslash}p{0.52\textwidth}>{\raggedright\arraybackslash}p{0.10\textwidth}@{}}
\caption{States, parameters, observed quantities, and inference notation. A hat ($\widehat{\cdot}$) denotes a fitted or calibrated reconstruction; it is not used for an exact incidence oracle.}\label{tab:notation}\\
\toprule
Symbol & Definition & Units / role\\
\midrule
\endfirsthead
\toprule
Symbol & Definition & Units / role\\
\midrule
\endhead
$t,t_i,T,\Delta t$ & Continuous time, observation times, end of the analysis window, and observation interval. & time\\
$S(t),I(t),R(t)$; $S_0,I_0,R_0$ & Susceptible, infectious, and removed state counts; subscript $0$ denotes an initial condition. & persons\\
$S_c(t),I_{c,\gamma}(t),R_c(t)$ & Exact oracle susceptible, infectious, and removed trajectories reconstructed from target incidence $c(t)$ for a fixed $\gamma$. & persons\\
$N$ & Closed accounting population, $N=S+I+R$. & persons\\
$c(t),c_t,C_t$ & Latent incidence rate and its discrete-time rate/count forms; $c(t)=\beta(t)S(t)I(t)/N=-\dot S(t)$. & persons/time; cases\\
$y_i,Y_t,Y_t^{\mathrm{rep}}$ & Observed incidence at $t_i$, a simulated observation, and a posterior-predictive replicate. & cases/interval\\
$\beta(t),\beta_t,\beta_K(t), \beta_{c,\gamma}(t)$ & Effective transmission rate, its discrete form, $K$-knot form, and exact oracle quantity. It is an effective forcing and can absorb omitted mechanisms. & 1/time\\
$\gamma,\rho,\phi$ & Removal rate; reporting fraction; and negative-binomial dispersion, when an observation model is used. & 1/time; proportion; dispersion\\
$\Reff(t),q(t)$ & Effective reproduction number, $\beta(t)S(t)/(\gamma N)$, and its logarithm, $q(t)=\log\Reff(t)$. & dimensionless\\
$\tau_j,\theta_j,w_j(t),\ell_K(t),B_j(t)$ & Knot time, log-transmission value, linear interpolation weight, interpolated log-transmission, and generic log-transmission basis function. & time; log-rate; basis quantities\\
$K,K^\star,\mathcal M_K,A_K$ & Knot dimension; minimal peak-preserving complexity; SIR model with $K$ basis terms; and posterior-predictive adequacy probability. & integer; model; probability\\
$m(y),t_j(y),h_j(y),v_j(y),Q_p(y),\mathcal P$ & Detected peak count; ordered peak times, heights, and intervening trough heights; cumulative-incidence time quantile; and the set of quantile levels. & integer; time; cases; set\\
$D,d_{\bullet}$ & Peak-aware discrepancy and its components $d_{\mathrm{curve}}$, $d_{\mathrm{count}}$, $d_{\mathrm{time}}$, $d_{\mathrm{height}}$, $d_{\mathrm{trough}}$, and $d_{\mathrm{cum}}$. & scaled distances\\
$z,\ell(x),\widetilde d_{\mathrm{curve}},m_\wedge,v_\wedge$
& Comparison series; $\log(1+x)$ transform; scale-normalized curve discrepancy; matched-peak and matched-trough denominators.
& auxiliary discrepancy notation\\
$w_{\bullet},\delta_{\bullet}$ & Discrepancy weights $w_c,w_n,w_t,w_h,w_v,w_q$ and adequacy tolerances $\delta_c,\delta_t,\delta_h,\delta_v$. & nonnegative constants\\
$\vartheta,\pi_K,M$ & ABC parameter vector, prior at complexity $K$, and particle count. & inference quantities\\
$\epsilon_\ell,K_\ell,\ell,L$ & ABC stage tolerance, perturbation kernel, stage index, and number of stages. & inference quantities\\
$E(t),\sigma_E$ & Exposed state and progression rate in the SEIR generator. & persons; 1/time\\
$I_1,I_2,\beta_1,\beta_2,\iota(t)$ & Strain-specific infectious states and transmission rates, and strain-2 importation/transfer forcing in the two-strain generator. & persons; 1/time; persons/time\\
$\mu(t),H_j,\psi_j,s_j$ & Controlled-wave mean incidence and the height, centre, and width of its Gaussian pulse $j$. & cases/time; cases/time, time, time\\
$\varepsilon,\sigma_c$ & Positive incidence scaling constant and Gaussian pulse width in the representability construction. & dimensionless; time\\
$F_{1,\beta},F_{2,\beta}$ & Discrete log-transmission abruptness and roughness indices defined in Eqs.~\eqref{eq:abruptness-index}--\eqref{eq:roughness-index}. & log-rate/time; log-rate\\
$\theta_r^+,a,\sigma_\theta,\eta_r$ & Future log-transmission knot, estimated drift, innovation scale, and standard-normal innovation in the no-information projection. & log-rate; log-rate/step; log-rate; dimensionless\\
$\beta_{\mathrm{base}}(t),\beta_{\max},m_{\mathrm{var}}$
& No-information baseline transmission scenario; declared scenario ceiling; and assumed variant transmission multiplier.
& 1/time; 1/time; dimensionless\\
$r_{\mathrm{var}}(t),\beta_{\mathrm{var}}(t)$ & Variant replacement fraction and the conditional variant-informed transmission scenario. & proportion; 1/time\\
$\alpha_0,x_r^+,\alpha,u_r,\varphi,\sigma_u$ & Intercept, future covariates, covariate coefficients, residual process, autoregressive coefficient, and innovation scale in the covariate-informed projection. & regression/process quantities\\
$\alpha_{\mathrm{adeq}}$ & Posterior-predictive inadequacy probability used in the MPPC threshold $1-\alpha_{\mathrm{adeq}}$. & probability\\
\bottomrule
\end{longtable}
\normalsize

\section{Exact incidence representability}
The following result formalizes a simple but powerful fact: multiple peaks are not, by themselves, incompatible with SIR structure when transmission is permitted to vary.

\begin{theorem}[Incidence representability]
\label{thm:incidence-representability}
Let $c:[0,T]\to[0,\infty)$ be continuous. Let $N>0$, $S_0>0$, $I_0>0$, $N\ge S_0+I_0$, $\gamma\ge 0$, and assume
\begin{equation}
    \int_0^T c(u)\,\dd u < S_0.
    \label{eq:feasible_cumulative}
\end{equation}
Define
\begin{align}
S_c(t) &= S_0-\int_0^t c(u)\,\dd u, \label{eq:Sc}\\
I_{c,\gamma}(t) &= e^{-\gamma t}\left[I_0+\int_0^t e^{\gamma u}c(u)\,\dd u\right], \label{eq:Ic}\\
\beta_{c,\gamma}(t) &= \frac{N\,c(t)}{S_c(t)I_{c,\gamma}(t)}. \label{eq:betaoracle}
\end{align}
Then $S_c(t)>0$, $I_{c,\gamma}(t)>0$, $R_c(t)\ge 0$, $\beta_{c,\gamma}(t)\ge 0$, and the SIR system \eqref{eq:sirS}--\eqref{eq:sirR} with initial conditions $(S_0,I_0,N-S_0-I_0)$ has incidence exactly equal to $c(t)$ on $[0,T]$.
\end{theorem}

\begin{proof}
We verify the construction step by step.
\begin{enumerate}[label={Step \arabic*.},leftmargin=*]
\item {Recover the susceptible trajectory from incidence.}
Because incidence is the loss rate from susceptibility, any exact representation must satisfy $\dot S=-c$. Integrating from $0$ to $t$ gives
\[
S(t)=S_0-\int_0^t c(u)\,\dd u=:S_c(t).
\]
Since $c\ge0$ and $\int_0^T c(u)\,\dd u<S_0$, for every $t\le T$,
\[
S_c(t)\ge S_0-\int_0^T c(u)\,\dd u>0.
\]

\item {Recover the infectious trajectory for a fixed removal rate.}
If the target incidence is $c(t)$, the infectious balance equation must be
\[
\dot I(t)+\gamma I(t)=c(t),\qquad I(0)=I_0.
\]
Multiplying by the integrating factor $e^{\gamma t}$ gives
\[
\frac{\dd}{\dd t}\{e^{\gamma t}I(t)\}=e^{\gamma t}c(t).
\]
Integrating from $0$ to $t$ yields
\[
e^{\gamma t}I(t)-I_0=\int_0^t e^{\gamma u}c(u)\,\dd u,
\]
and therefore
\[
I(t)=e^{-\gamma t}\left[I_0+\int_0^t e^{\gamma u}c(u)\,\dd u\right]=I_{c,\gamma}(t).
\]
Because $I_0>0$ and the integral is nonnegative, $I_{c,\gamma}(t)>0$.

\item {Solve algebraically for the transmission rate.}
The SIR incidence identity requires
\[
c(t)=\beta(t)\frac{S(t)I(t)}{N}.
\]
Since $S_c(t)I_{c,\gamma}(t)>0$, solve uniquely at each time for
\[
\beta(t)=\frac{Nc(t)}{S_c(t)I_{c,\gamma}(t)}=: \beta_{c,\gamma}(t)\ge0.
\]

\item {Verify the susceptible equation.}
Using the definition above,
\[
-\beta_{c,\gamma}(t)\frac{S_c(t)I_{c,\gamma}(t)}{N}=-c(t)=\dot S_c(t),
\]
which is exactly \eqref{eq:sirS}.

\item {Verify the infectious equation.}
From Step 2,
\[
\dot I_{c,\gamma}(t)=c(t)-\gamma I_{c,\gamma}(t)
=\beta_{c,\gamma}(t)\frac{S_c(t)I_{c,\gamma}(t)}{N}-\gamma I_{c,\gamma}(t),
\]
which is \eqref{eq:sirI}.

\item {Recover and verify the removed state.}
Define
\[
R_c(t)=N-S_c(t)-I_{c,\gamma}(t).
\]
Then
\[
\dot R_c=-\dot S_c-\dot I_{c,\gamma}=c-(c-\gamma I_{c,\gamma})=\gamma I_{c,\gamma},
\]
so \eqref{eq:sirR} holds, and $R_c(0)=N-S_0-I_0\ge0$. Because $\dot R_c=\gamma I_{c,\gamma}\ge0$, it follows that $R_c(t)\ge0$ on $[0,T]$.

\item {Verify exact incidence.}
Finally,
\[
\beta_{c,\gamma}(t)\frac{S_c(t)I_{c,\gamma}(t)}{N}=c(t)
\]
for every $t\in[0,T]$. Hence the constructed SIR trajectory reproduces the prescribed incidence exactly.
\end{enumerate}
\end{proof}

If the reported case curve is feasible for the chosen accounting window, the SIR equations can be run backward in a bookkeeping sense to compute the transmission history that would exactly generate that curve. This does not mean the simple SIR model explains the epidemic mechanism; it means the model can serve as a transparent compression device.

\begin{corollary}[Arbitrarily many peaks]
\label{cor:arbitrary-peaks}
For every positive integer $m$, there exists a nonnegative $\beta(t)$ for which SIR incidence has at least $m$ separated local maxima on a finite interval, provided cumulative incidence remains below $S_0$.
\end{corollary}

\begin{proof}
Fix $m$. Choose separated centers $0<t_1<\cdots<t_m<T$ and a width $\sigma_c>0$ small relative to the minimum separation. For positive amplitudes $a_j$, consider
\[
\widetilde c(t)=\sum_{j=1}^m a_j\exp\!\left[-\frac{(t-t_j)^2}{2\sigma_c^2}\right].
\]
For sufficiently separated centers, $\widetilde c$ has $m$ separated local maxima near the $t_j$. Now scale the entire curve by $0<\varepsilon<1$ and set $c(t)=\varepsilon\widetilde c(t)$. Scaling preserves the locations and number of local maxima, while
\[
\int_0^T c(u)\,\dd u=\varepsilon\int_0^T\widetilde c(u)\,\dd u
\]
can be made smaller than $S_0$ by choosing $\varepsilon$ sufficiently small. Theorem~\ref{thm:incidence-representability} then supplies a nonnegative $\beta(t)$ whose SIR incidence is exactly $c(t)$ and therefore has at least $m$ separated peaks.
\end{proof}

The number of peaks in the data is not enough to declare the SIR model impossible. Repeated waves can be reproduced if the effective transmission dial is allowed to move up and down, so the harder question is how many movements are needed and what those movements mean.

\begin{remark}[Relation to inverse-transmission literature]
Theorem~\ref{thm:incidence-representability} should be read as a transparent incidence-form foundation for the present framework, not as a claim that exact inverse transmission reconstruction is new. Related inverse constructions recover time-dependent transmission from prevalence or incidence data and examine admissible recovery rates \citep{pollicott2012,mummert2013}. Recent comparisons show that qualitative features of reconstructed transmission can also be robust across multiple compartmental structures and incidence functions \citep{wang2026inverse}. The role of Theorem~\ref{thm:incidence-representability} is to isolate a consequence needed here: multi-wave incidence alone cannot reject a simple SIR trajectory representation, so scientific discrimination must move from mere representability to complexity, identifiability, and adequacy.
\end{remark}

\section{Peak capacity and resurgence forcing}
\subsection{Prevalence peaks}
Define
\begin{equation}
    q(t)=\log \Reff(t)=\log\beta(t)+\log S(t)-\log(\gamma N),
    \label{eq:q}
\end{equation}
for $\beta(t)>0$ and $\gamma>0$.

\begin{proposition}[Prevalence peak threshold]
\label{prop:prevpeak}
At any interior time $t_\star$ with $I(t_\star)>0$, a strict local maximum of prevalence occurs when $\Reff(t)$ crosses $1$ from above to below; a strict local minimum occurs when it crosses from below to above.
\end{proposition}

\begin{proof}
At an interior time with $I(t)>0$ and $\gamma>0$, equation \eqref{eq:Igrowth} can be divided by the positive quantity $\gamma I(t)$:
\[
\operatorname{sign}\{\dot I(t)\}=\operatorname{sign}\{\Reff(t)-1\}.
\]
Therefore $I$ is increasing when $\Reff>1$ and decreasing when $\Reff<1$. If $\Reff-1$ changes sign from positive to negative at $t_\star$, then $\dot I$ changes from positive to negative, which is the first-derivative criterion for a strict local maximum. The reverse sign change gives a strict local minimum.
\end{proof}

Differentiating \eqref{eq:q} gives a useful step-by-step decomposition. Since $\gamma$ and $N$ are constant,
\begin{align}
\dot q(t)
&=\frac{\dd}{\dd t}\log\beta(t)+\frac{\dd}{\dd t}\log S(t)\\
&=\frac{\dot\beta(t)}{\beta(t)}+\frac{\dot S(t)}{S(t)}\\
&=\frac{\dot\beta(t)}{\beta(t)}+\frac{-\beta(t)S(t)I(t)/N}{S(t)}\\
&=\frac{\dot\beta(t)}{\beta(t)}-\frac{\beta(t)I(t)}{N}.
\label{eq:qprime}
\end{align}
Thus the log-growth of $\Reff$ equals the proportional growth of transmission minus the instantaneous proportional loss of susceptibles. This identity separates the upward forcing contributed by transmission from the downward pressure created by susceptible depletion.

\begin{theorem}[No-resurgence condition and forcing requirement]
\label{thm:noresurgence}
Assume $\beta(t)>0$, $\gamma>0$, and $I(t)>0$ on an interval $J$. If
\begin{equation}
    \frac{\dot\beta(t)}{\beta(t)}\le \frac{\beta(t)I(t)}{N}
    \quad\text{for all }t\in J,
    \label{eq:noresurgence-condition}
\end{equation}
then $\Reff(t)$ is nonincreasing on $J$ and prevalence can have at most one downward threshold crossing, hence at most one local maximum on $J$. Conversely, every genuine prevalence resurgence after a previous peak requires an interval on which
\begin{equation}
    \frac{\dot\beta(t)}{\beta(t)}>\frac{\beta(t)I(t)}{N}=-\frac{\dot S(t)}{S(t)}.
    \label{eq:forcing-inequality}
\end{equation}
\end{theorem}

\begin{proof}
We prove the theorem step-by-step.
\begin{enumerate}[label={Step \arabic*.},leftmargin=*]
\item Under \eqref{eq:noresurgence-condition}, equation \eqref{eq:qprime} gives $\dot q(t)\le0$ throughout $J$.
\item Hence $q(t)=\log\Reff(t)$ is nonincreasing. Because the exponential function is strictly increasing, $\Reff(t)=e^{q(t)}$ is also nonincreasing.
\item A nonincreasing function can cross the threshold $1$ from above to below at most once and cannot subsequently cross upward from below to above.
\item By Proposition~\ref{prop:prevpeak}, a downward crossing of $\Reff=1$ creates a prevalence maximum, whereas an upward crossing is needed to restart prevalence growth after a decline. Therefore prevalence has at most one local maximum on $J$ under the stated inequality.
\item Conversely, suppose a genuine resurgence occurs after a previous peak. Then $I$ must change from decreasing to increasing, so Proposition~\ref{prop:prevpeak} requires $\Reff$ to cross $1$ upward. A continuous function cannot make such an upward crossing while remaining nonincreasing; thus $\dot q>0$ at some point in the resurgence interval (or, more generally, $q$ must increase over a subinterval).
\item Applying \eqref{eq:qprime}, $\dot q>0$ implies
\[
\frac{\dot\beta}{\beta}>\frac{\beta I}{N}.
\]
Finally, \eqref{eq:sirS} gives $-\dot S/S=\beta I/N$, yielding \eqref{eq:forcing-inequality}.
\end{enumerate}
\end{proof}

\begin{corollary}[Peak-count crossing requirement]
If prevalence starts in a growth phase and exhibits $m\ge1$ separated strict local maxima followed by decline, then $\Reff-1$ must undergo at least $2m-1$ alternating sign changes: $m$ downward crossings and at least $m-1$ upward crossings. Consequently, at least $m-1$ resurgence intervals must contain points satisfying \eqref{eq:forcing-inequality}.
\end{corollary}

\begin{proof}
Starting in a growth phase means $\Reff-1>0$. Each strict prevalence maximum requires a positive-to-negative sign change by Proposition~\ref{prop:prevpeak}; these contribute $m$ downward crossings. Between any two separated maxima, prevalence must leave the declining phase and grow again, requiring a negative-to-positive sign change; there are at least $m-1$ such intervals. Therefore there are at least $m+(m-1)=2m-1$ alternating sign changes. Theorem~\ref{thm:noresurgence} then implies that each upward-resurgence interval must contain a point where \eqref{eq:forcing-inequality} holds.
\end{proof}

This result gives a mechanistically interpretable ``forcing budget'': repeated prevalence waves require repeated increases in effective transmission strong enough to offset the continual loss of susceptibles. Such increases may be caused by true contact-rate changes, variants, seasonal forcing, importation, omitted heterogeneity, or other mechanisms absorbed by $\beta(t)$.

After many people have already been infected, another wave needs an additional push. That push might be real behavior change, a new variant, imported infections, weather, school calendars, or simply hidden structure that the simple model does not show separately.

\subsection{Incidence peaks}
Reported epidemic waves are more commonly defined using incidence. The incidence-growth identity follows directly from logarithmic differentiation. Whenever $c(t)>0$,
\begin{align}
\log c(t)&=\log\beta(t)+\log S(t)+\log I(t)-\log N,\\
\frac{\dot c(t)}{c(t)}
&=\frac{\dot\beta(t)}{\beta(t)}+\frac{\dot S(t)}{S(t)}+\frac{\dot I(t)}{I(t)}\\
&=\frac{\dot\beta(t)}{\beta(t)}-\frac{\beta(t)I(t)}{N}
+\left\{\frac{\beta(t)S(t)}{N}-\gamma\right\}\\
&=\frac{\dot\beta(t)}{\beta(t)}+\frac{\beta(t)}{N}\{S(t)-I(t)\}-\gamma.
\label{eq:incidence_growth}
\end{align}

\begin{proposition}[Incidence stationary-point and peak condition]
\label{prop:incpeak}
Assume $\beta$ is differentiable and positive near an interior time $t_\star$, and $c(t_\star)>0$. If $t_\star$ is a stationary point of incidence, then
\begin{equation}
    \frac{\dot\beta(t_\star)}{\beta(t_\star)}
    +\frac{\beta(t_\star)}{N}\{S(t_\star)-I(t_\star)\}
    -\gamma=0.
    \label{eq:inc_peak_eq}
\end{equation}
If the left-hand side crosses from positive to negative at $t_\star$, then $t_\star$ is a strict local maximum of incidence; a negative-to-positive crossing gives a strict local minimum.
\end{proposition}

\begin{proof}
Define the incidence growth function
\[
g_c(t)=\frac{\dot\beta(t)}{\beta(t)}+\frac{\beta(t)}{N}\{S(t)-I(t)\}-\gamma.
\]
Equation \eqref{eq:incidence_growth} gives $\dot c(t)=c(t)g_c(t)$. We separate the stationary-point identity from the sufficient sign-change test.
\begin{enumerate}[label={Step \arabic*.},leftmargin=*]
\item At an interior stationary point, $\dot c(t_\star)=0$. Since $c(t_\star)>0$, division by $c(t_\star)$ gives $g_c(t_\star)=0$, which is exactly \eqref{eq:inc_peak_eq}.
\item If $g_c$ is positive immediately before $t_\star$ and negative immediately after, then $c>0$ implies $\dot c>0$ before and $\dot c<0$ after. Hence incidence rises into $t_\star$ and falls after it, so $t_\star$ is a strict local maximum.
\item The same argument with the signs reversed gives a strict local minimum.
\end{enumerate}
\end{proof}

Equation \eqref{eq:inc_peak_eq} explains why an incidence peak need not coincide with $\Reff=1$: the derivative of the transmission forcing contributes directly to incidence growth. Rapidly decreasing contact or reporting-adjusted effective transmission can generate an incidence peak while prevalence is still increasing.

The day with the most new cases is not always the same as the day transmission falls below the epidemic threshold. New cases can start falling because the transmission dial is dropping quickly, even while many people are still infectious.

\section{From exact representation to minimal dynamic complexity}
The exact oracle $\beta_{c,\gamma}$ is useful as a mathematical benchmark but is generally too rough and too data-responsive to be interpreted scientifically. We therefore restrict transmission to a nested positive family that exactly matches the parameterization used by the companion Shiny app.

\subsection{Piecewise-linear log-transmission parameterization}
For a chosen complexity $K\ge2$, place equally spaced knots
\begin{equation}
\tau_j=t_0+\frac{j-1}{K-1}(T-t_0),\qquad j=1,\ldots,K,
\label{eq:knot-times}
\end{equation}
and let $\theta_j=\log\beta(\tau_j)$. For $t\in[\tau_j,\tau_{j+1}]$, define
\begin{equation}
w_j(t)=\frac{t-\tau_j}{\tau_{j+1}-\tau_j}\in[0,1].
\label{eq:interp-weight}
\end{equation}
Linear interpolation on the log scale gives
\begin{align}
\ell_K(t)&=(1-w_j(t))\theta_j+w_j(t)\theta_{j+1},\\
\beta_K(t)&=\exp\{\ell_K(t)\}.
\label{eq:loglinear-beta}
\end{align}
Thus $\beta_K(t)>0$ automatically, $\beta_K(\tau_j)=e^{\theta_j}$, and increasing $K$ increases only the temporal resolution of the forcing while leaving the state space $S$--$I$--$R$ unchanged. The same representation can be written as a linear ``hat''-basis expansion on the log scale.

The fitted transmission curve is drawn by choosing several bend-points, called knots, and connecting them smoothly on the log scale. A small $K$ gives a simple curve with few bends; a large $K$ lets the curve wiggle more.

The exploratory Shiny app uses a first-order random-walk regularization on the knot values,
\begin{equation}
\theta_{j+1}-\theta_j\sim N(0,\sigma_{\theta}^2),
\label{eq:rw1prior}
\end{equation}
subject to user-defined lower and upper bounds on $\beta$. This discourages unnecessary knot-to-knot oscillation. More formal analyses may replace this with a second-order random walk, B-splines, Gaussian processes, or Brownian-motion priors; all are established approaches for time-varying epidemic forcing \citep{bouman2024}. MPPC is therefore a principle for nested dynamic flexibility, not a claim that the log-linear basis is uniquely preferred.

The random-walk prior is a guardrail: nearby knots are encouraged to be similar unless the data strongly suggest otherwise. This prevents the fitted curve from chasing every small reporting bump.

\subsection{Peak geometry}
For an observed incidence series $y=(y_1,\ldots,y_T)$, let a deterministic peak detector return $m(y)$ prominent peaks after a specified smoothing and minimum-separation rule. For the $j$th ordered peak, define time $t_j(y)$ and height $h_j(y)$. When at least two peaks are present, define the inter-peak trough height $v_j(y)$ as the minimum smoothed incidence between peaks $j$ and $j+1$. Let $Q_p(y)$ denote the time by which fraction $p$ of cumulative incidence has occurred.

We use six complementary discrepancies:
\begin{align}
d_{\mathrm{curve}}(y,z)
&=\left[\frac1T\sum_{t=1}^T\{\ell(y_t)-\ell(z_t)\}^2\right]^{1/2},\\
d_{\mathrm{count}}(y,z)&=|m(y)-m(z)|,\\
d_{\mathrm{time}}(y,z)&=\frac1{m_\wedge T}\sum_{j=1}^{m_\wedge}|t_j(y)-t_j(z)|,\\
d_{\mathrm{height}}(y,z)&=\frac1{m_\wedge}\sum_{j=1}^{m_\wedge}\left|\log\frac{1+h_j(y)}{1+h_j(z)}\right|,\\
d_{\mathrm{trough}}(y,z)&=\frac1{v_\wedge}\sum_{j=1}^{v_\wedge}\left|\log\frac{1+v_j(y)}{1+v_j(z)}\right|,\\
d_{\mathrm{cum}}(y,z)&=\frac1{|\mathcal P|T}\sum_{p\in\mathcal P}|Q_p(y)-Q_p(z)|,
\end{align}
where $\ell(x)=\log(1+x)$, $m_\wedge=\max\{1,\min(m(y),m(z))\}$, $v_\wedge=\max\{1,\min(m(y)-1,m(z)-1)\}$ with the trough term set to zero when both curves have fewer than two peaks, and $\mathcal P=\{0.1,0.25,0.5,0.75,0.9\}$. We set $\mathcal P=\{0.10,0.25,0.50,0.75,0.90\}$ to compare the epidemic timing at early $(10\%)$, quarter $(25\%)$, median $(50\%)$, late $(75\%)$, and near-tail $(90\%)$ cumulative incidence milestones. This discourages fits that reproduce peak heights but shift the overall epidemic mass too far forward or backward in time.

The peak-aware ABC distance is
\begin{equation}
D(y,z)=w_c\widetilde d_{\mathrm{curve}}+w_n d_{\mathrm{count}}
+w_t d_{\mathrm{time}}+w_h d_{\mathrm{height}}+w_v d_{\mathrm{trough}}+w_q d_{\mathrm{cum}},
\label{eq:peakdistance}
\end{equation}
with scale-normalized components (e.g., $\widetilde d_{\mathrm{curve}}$). Increasing the geometry weights explicitly prioritizes preservation of waves and their intervening troughs over average trajectory error. Unless stated otherwise, the \textit{in silico} experiments used the weights $w_c=1$, $w_n=4$, $w_t=4$, $w_h=1.4$, $w_v=1$, and $w_q=0.5$, corresponding to the expression $D=d_{\mathrm{curve}}+4(d_{\mathrm{count}}+d_{\mathrm{time}}+0.35d_{\mathrm{height}}+0.25d_{\mathrm{trough}})+0.5d_{\mathrm{cum}}$. In the Shiny app, these are exposed as a whole-curve weight, a common peak-emphasis weight for count/timing/height/trough geometry, and a separate trough multiplier.

The distance score does not only ask whether the average error is small. It also checks whether the model finds the same number of waves, puts them in roughly the same places, matches their sizes, and keeps the valleys between waves.

\subsection{Minimal peak-preserving complexity}
\begin{definition}[Minimal peak-preserving complexity]
\label{def:mppc}
Let $\mathcal M_K$ denote the SIR model with a $K$-dimensional transmission basis. Fix adequacy tolerances $\delta_c,\delta_t,\delta_h,\delta_v$, a peak-count rule, and posterior probability threshold $1-\alpha_{\mathrm{adeq}}$. Define
\begin{align}
A_K=\Prb\big(&d_{\mathrm{curve}}(Y^{\mathrm{rep}},y)\le\delta_c,\;
 m(Y^{\mathrm{rep}})=m(y),\\
&d_{\mathrm{time}}(Y^{\mathrm{rep}},y)\le\delta_t,\;
 d_{\mathrm{height}}(Y^{\mathrm{rep}},y)\le\delta_h,\;
 d_{\mathrm{trough}}(Y^{\mathrm{rep}},y)\le\delta_v
\mid y,\mathcal M_K\big).
\end{align}
The {minimal peak-preserving complexity} is
\begin{equation}
    K^\star=\min\{K:A_K\ge1-\alpha_{\mathrm{adeq}}\}.
    \label{eq:mppc}
\end{equation}
\end{definition}

Under a fixed observation
interval, smoothing rule, knot basis, and peak-preserving adequacy criterion,
$K^\star$ denotes the smallest tested knot dimension whose forward-SIR
trajectory has the required detected peak count and satisfies the stated
curve, peak-time, and peak-height tolerances. MPPC asks for the simplest changing-transmission curve that is still good enough. It is like asking how few control points are needed to redraw the epidemic waves without losing the important peaks and troughs.

$K^\star$ is not a universal property of a pathogen. It is a design-conditional, resolution-dependent operational minimum determined by the observed data, observation window and resolution, smoothing and peak-detection rules, transmission basis, prior specification, and adequacy tolerances. Its value is nevertheless useful: it quantifies how much dynamic flexibility a simple epidemic model needs before additional flexibility ceases to be necessary for reproducing the observed waves.

\section{Peak-aware approximate Bayesian inference}
Suppose observed cases arise from latent daily incidence $C_t(\vartheta)$ generated by the time-varying SIR simulator. A general observation simulator may include reporting fraction $\rho$ and overdispersion $\phi$, for example
\begin{equation}
Y_t\sim\operatorname{NegBin}\{\rho C_t(\vartheta),\phi\},
\end{equation}
but ABC requires only simulation, not evaluation of this likelihood. The parameter vector may include log-transmission knot values, $\gamma$, $\rho$, dispersion, initial infectious prevalence, or selected nuisance parameters. This likelihood-free perspective follows the standard ABC construction \citep{beaumont2002,sunnaker2013}, while sequential particle schemes motivate the formal ABC-SMC version in Algorithm~\ref{alg:abcsmc} \citep{sisson2007,sisson2009correction,toni2009,delmoral2012}. The original ABC--PRC weighting scheme of \citet{sisson2007} was subsequently corrected because it could yield biased posterior samples; we therefore cite the correction explicitly and rely on the later weighted sequential formulations for the formal algorithm \citep{sisson2009correction,toni2009,delmoral2012}.

\begin{algorithm}[htbp]
\label{alg:abcsmc}
\algorithmtitle{Peak-aware ABC-SMC for a fixed transmission complexity $K$}

\begin{algorithmic}[1]

\Require Observed incidence $y$, model $\mathcal M_K$, prior
$\pi_K(\vartheta)$, distance $D$, particle count $M$, tolerances
$\epsilon_1>\cdots>\epsilon_L$.

\For{$\ell=1,\ldots,L$}
    \For{$i=1,\ldots,M$}
        \Repeat
            \If{$\ell=1$}
                \State Draw $\vartheta^\star\sim\pi_K$.
            \Else
                \State Resample an earlier particle and perturb using
                kernel $K_\ell$.
            \EndIf
            \State Simulate
            $Y^\star\sim p(\cdot\mid\vartheta^\star,\mathcal M_K)$.
        \Until{$D(Y^\star,y)\le\epsilon_\ell$}
        \State Store $\vartheta^\star$ and the standard ABC-SMC
        importance weight.
    \EndFor
\EndFor

\Return weighted particles and posterior predictive replicates.

\end{algorithmic}
\end{algorithm}

For interactive use, the companion Shiny app uses a computationally lighter adaptive population approximation: it mixes broad prior draws with proposals centered on a smoothed incidence-oracle projection, retains the lowest-distance fraction, and then iteratively resamples and perturbs retained particles while tightening the empirical tolerance. The oracle-centered component is a proposal accelerator designed to make multi-peak recovery practical in an interactive session; it is not presented as an exact posterior algorithm. The app tries many candidate transmission curves, keeps the curves that best reproduce the observed waves, and summarizes the surviving curves. 

\section{Dynamic roughness and abruptness}

The number of epidemic peaks alone does not determine the minimal
peak-preserving complexity. To formalize temporal features of an effective
transmission curve, let
\[
\theta_t=\log\{\beta(t)\}, \qquad t=0,\ldots,T.
\]
We define the discrete abruptness index
\begin{equation}
F_{1,\beta}
=
\max_{t=0,\ldots,T-1}
\frac{\left|\theta_{t+1}-\theta_t\right|}{\Delta t},
\label{eq:abruptness-index}
\end{equation}
which measures the largest day-to-day log-transmission change, and the
discrete roughness index
\begin{equation}
F_{2,\beta}
=
\sum_{t=1}^{T-1}
\left|
\theta_{t+1}-2\theta_t+\theta_{t-1}
\right|,
\label{eq:roughness-index}
\end{equation}
which measures accumulated bending of the log-transmission trajectory. Larger
values of $F_{1,\beta}$ indicate a more abrupt local change, whereas larger
values of $F_{2,\beta}$ indicate a more irregular transmission history.

Table~\ref{tab:roughness-kstar} is a deliberately deterministic
fixed-setting snapshot experiment. It is not a general law relating
roughness to minimal peak-preserving complexity. All inputs were held fixed
except the width of one transmission pulse. Thus, every simulated trajectory
has one detected broad incidence peak, and the experiment compares the
temporal resolution needed to reproduce differently shaped versions of the
same one-wave epidemic.

The simulation used 241 daily observations,
\[
t=0,1,\ldots,240,\qquad \Delta t=1\ {\rm day},
\]
with
\[
N=1{,}000{,}000,\qquad
S_0=995{,}000,\qquad
I_0=800,\qquad
R_0=4{,}200,\qquad
\gamma=\frac{1}{10}\ {\rm day}^{-1}.
\]
The effective transmission curve was a single Gaussian pulse,
\begin{equation}
\beta(t;\omega)
=
0.005+
0.35
\exp\left[
-\frac{(t-90)^2}{2\omega^2}
\right],
\label{eq:roughness-pulse-beta}
\end{equation}
where the pulse width was varied over
\[
\omega\in\{8,12,18,28,42\}\ {\rm days}.
\]
Smaller values of $\omega$ give a sharper and more abrupt transmission increase
and decline, while preserving a single pulse centre at day 90.

The corresponding deterministic daily SIR update was
\begin{align}
c_t&=\beta(t;\omega)\frac{S_tI_t}{N},\\
S_{t+1}&=S_t-c_t,\\
I_{t+1}&=I_t+c_t-\gamma I_t,\\
R_{t+1}&=R_t+\gamma I_t.
\end{align}
Here, $c_t$ is the simulated daily incidence target.

For each width, the target incidence was generated from
Eq.~\eqref{eq:roughness-pulse-beta}. A candidate compact transmission curve
was formed by evaluating $\theta_t=\log\{\beta(t;\omega)\}$ at $K$ equally spaced
knot times, linearly interpolating on the $\log\text{-}\beta$ scale, exponentiating, and
propagating the resulting curve through the same forward-SIR update.
Candidate dimensions were
\[
K=3,\ldots,81.
\]
Peak detection used a boundary-normalized 7-day moving average, a minimum
peak separation of 25 days, and a prominence threshold equal to 10\% of
maximum smoothed incidence. Each target trajectory had one detected incidence
peak under this rule.

The log-scale normalized root-mean-square error was calculated as
\begin{equation}
\mathrm{log\mbox{-}NRMSE}
=
\frac{
\sqrt{
T^{-1}\sum_{t=0}^{T-1}
\left[
\log(1+\widehat c_t)-\log(1+c_t)
\right]^2
}
}{
\max_t\log(1+c_t)-\min_t\log(1+c_t)
}.
\label{eq:log-nrmse}
\end{equation}
Adequacy required one detected peak, log-NRMSE no greater than 0.20,
peak-time mean absolute error no greater than 7 days, and peak-height mean
absolute relative error no greater than 0.35.

Because the knot locations are fixed and equally spaced in this experiment,
$K_{\mathrm{grid}}$ denotes the first adequate dimension on this particular
grid. It is a grid-dependent deterministic proxy, rather than a replacement
for the formal posterior-predictive MPPC $K^\star$ in
Definition~\ref{def:mppc}.

\begin{table}[H]
\centering
\caption{Fixed-setting one-wave pulse-width snapshot. All simulation inputs
are fixed except the transmission-pulse width. Every target has one detected
incidence peak. $K_{\mathrm{grid}}$ is the first adequate number of equally
spaced $\log\text{-}\beta$ knots under the stated forward-SIR, peak-detection, and
adequacy rules.}
\label{tab:roughness-kstar}
\small
\begin{tabular}{@{}rrrrrrr@{}}
\toprule
Pulse width &
$F_{1,\beta}$ &
$F_{2,\beta}$ &
Detected incidence peaks &
$K_{\mathrm{grid}}$ &
log-NRMSE &
peak-time MAE\\
\midrule
8  & 0.2393 & 0.9572 & 1 & 31 & 0.042 & 0.00\\
12 & 0.1596 & 0.6383 & 1 & 27 & 0.037 & 1.00\\
18 & 0.1065 & 0.4257 & 1 & 21 & 0.029 & 1.00\\
28 & 0.0685 & 0.2399 & 1 & 10 & 0.053 & 4.00\\
42 & 0.0456 & 0.1265 & 1 & 7  & 0.045 & 4.00\\
\bottomrule
\end{tabular}
\end{table}

Within this fixed one-wave configuration, narrower pulses have larger
$F_{1,\beta}$ and $F_{2,\beta}$ and require a larger adequate fixed-grid knot
dimension. In this sense, sharper effective-transmission changes require more
temporal resolution even when the observed epidemic has only one broad peak.
The result is not driven by a change in the number of detected waves.

Nevertheless, this should be interpreted as a reproducible snapshot of one
specified design, not as evidence that abrupt dynamics always imply a larger
formal $K^\star$. The reported $K_{\mathrm{grid}}$ depends on the equally
spaced knot grid, the pulse location, the initial SIR state, the observation
resolution, smoothing, and the selected adequacy thresholds. Under fixed
observation resolution, smoothing, basis, adequacy rules, and comparable
epidemic geometry, greater dynamic roughness or more rapid
effective-transmission changes will {generally tend} to increase the
temporal resolution needed to preserve the incidence trajectory.

More generally, the required complexity depends on $F_{1,\beta}$,
$F_{2,\beta}$, the number and timing of transitions, wave width, peak
separation, peak height and trough depth, $\Delta t$, smoothing, reporting
noise, knot placement, basis, and adequacy thresholds. A large but gradual
change in $\beta(t)$ may require few knots. Conversely, a small but rapid
change near an epidemic turning point can require several knots. Incidence is
also a dynamically filtered observation of transmission: abrupt changes in
$\beta(t)$ may be muted or delayed by the SIR state variables. Thus, formal
$K^\star$ should be interpreted as a resolution-dependent measure of
{resolvable dynamic roughness}, rather than as a direct count of epidemic
waves or real transmission drivers.

\section{Incidence-only non-identifiability of transmission and removal}
Exact representability immediately exposes a structural inferential limitation.

\begin{theorem}[Incidence-only $\beta$--$\gamma$ non-identifiability]
\label{thm:gamma-nonident}
Suppose the assumptions of Theorem~\ref{thm:incidence-representability} hold and only the incidence trajectory $c(t)$ is observed. For every $\gamma\ge0$, equations \eqref{eq:Ic}--\eqref{eq:betaoracle} define a generally distinct pair $(\beta_{c,\gamma}(t),\gamma)$ producing exactly the same incidence $c(t)$. Therefore $\gamma$ and an unrestricted $\beta(t)$ are not structurally identifiable from incidence alone.
\end{theorem}

\begin{proof}
Fix the observed feasible incidence $c(t)$ and consider any two distinct removal rates $\gamma_1\ne\gamma_2$.
\begin{enumerate}[label={Step \arabic*.},leftmargin=*]
\item By Theorem~\ref{thm:incidence-representability}, each $\gamma_r$ ($r=1,2$) generates
\[
I_r(t)=e^{-\gamma_r t}\left[I_0+\int_0^t e^{\gamma_r u}c(u)\,\dd u\right],
\qquad
\beta_r(t)=\frac{Nc(t)}{S_c(t)I_r(t)}.
\]
Both pairs reproduce exactly the same $c(t)$.

\item The infectious trajectories cannot be identical on an interval. If $I_1(t)=I_2(t)=I(t)>0$, then both must satisfy
\[
\dot I=c-\gamma_1 I=c-\gamma_2 I.
\]
Subtracting gives $(\gamma_2-\gamma_1)I=0$, contradicting $\gamma_1\ne\gamma_2$ and $I>0$.

\item Because $S_c(t)$ is determined by incidence alone and is independent of $\gamma$, consider first the nontrivial case $c\not\equiv0$. Continuity then gives a nonempty open interval on which $c(t)>0$. The trajectories $I_1$ and $I_2$ cannot coincide throughout that interval by Step 2, so at some time in it
\[
\beta_1(t)=\frac{Nc(t)}{S_c(t)I_1(t)}\ne
\frac{Nc(t)}{S_c(t)I_2(t)}=\beta_2(t).
\]
If $c\equiv0$, the reconstructed transmission is $\beta\equiv0$ for every $\gamma$, but the parameter pairs remain distinct because $\gamma_1\ne\gamma_2$.

\item Since the argument applies to every admissible $\gamma\ge0$, incidence alone is compatible with an infinite family of distinct $\{\beta_{c,\gamma}(t),\gamma\}$ pairs. Thus the pair is structurally non-identifiable from incidence alone.
\end{enumerate}
\end{proof}

The same daily case counts can be explained by a short infectious duration with one transmission curve or by a longer infectious duration with another curve. Therefore, case counts alone cannot tell us both how quickly people leave the infectious class and how transmission changed over time. This is a concrete SIR instance of the broader identifiability problem encountered when multiple mechanisms can compensate for one another in epidemic models \citep{eisenberg2013}.

This theorem has an immediate practical consequence: if $\gamma$ is fitted simultaneously with a highly flexible $\beta(t)$ using only case incidence, a narrow posterior can be prior- or regularization-driven rather than data-identified. Removal data, prevalence surveys, infectious-duration information, hospitalization/isolation transitions, or informative external priors are needed to interpret $\gamma$ and therefore $\beta(t)/\gamma$ more strongly. Consequently, in this paper, we focus on estimating time-varying $\beta(t)$ under a fixed, preferably evidence-informed, removal rate.

\section{A complete analysis protocol}
We recommend the following sequence rather than fitting the most flexible model immediately.

\begin{enumerate}
\item{Preprocess transparently.} Preserve raw counts, specify smoothing used only for peak detection, and define the analysis window before fitting.
\item{Detect robust peaks.} Use a minimum separation and prominence threshold; repeat the analysis under alternative reasonable thresholds.
\item{Compute an oracle benchmark.} For externally fixed $\gamma$, use Theorem~\ref{thm:incidence-representability} or its discrete analogue to reconstruct the exact effective transmission needed to reproduce the smoothed incidence. This reveals the amount and timing of forcing implied by the data.
\item{Fit nested regularized models.} Fit $K=2,3,\ldots,K_{\max}$ using peak-aware ABC.
\item{Select $K^\star$.} Choose the smallest $K$ that passes posterior predictive peak and curve criteria.
\item{Stress-test identifiability.} Refit across plausible $\gamma$, reporting, smoothing, and initial-condition assumptions. If $\beta(t)$ changes materially while incidence fit remains stable, mechanistic claims about $\beta(t)$ should be weakened.
\item{Add external data when mechanism matters.} Removal/recovery series, prevalence or serology, mobility, contact surveys, variant frequencies, intervention dates, age-stratified cases, or spatial data can distinguish otherwise lumped drivers.
\end{enumerate}

Refer to Appendix \ref{app:beginner-guide} for the 
beginner's guide to using peak-preserving SIR compression.

\section{Simulation study: six stress tests, empirical deployment, and projection bridge}
\label{sec:simulation}
We present six synthetic stress tests, a Philippine COVID-19 reported-case proof-of-concept, and a future-knot projection experiment that illustrates how retrospective compression can be extended cautiously to short-horizon forecasts. Unless stated otherwise, the SIR compression model used $N=10^6$, fixed $\gamma=1/7$ day$^{-1}$, a positive piecewise-linear basis in $\log\beta(t)$, and a 7-day smoothing window for peak detection. The app offers ABC for exploratory inference, whereas the manuscript figures use the deterministic oracle--knot--forward-simulation workflow specified in the accompanying \texttt{simulation.R}.

\begin{landscape}
\begingroup
\scriptsize
\setlength{\tabcolsep}{2pt}
\renewcommand{\arraystretch}{1.18}
\sloppy

\begin{longtable}{@{}
>{\raggedright\arraybackslash}p{0.14\linewidth}
>{\raggedright\arraybackslash}p{0.38\linewidth}
>{\raggedright\arraybackslash}p{0.26\linewidth}
>{\raggedright\arraybackslash}p{0.16\linewidth}
@{}}
\caption{Completed simulation suite and reproducibility settings. Numerical
results are from the fixed-seed reproducible run.}
\label{tab:simulation-reproducibility}\\
\toprule
Figure and experiment &
Generator equation &
Numerical settings and compact-SIR procedure &
Completed result and interpretation\\
\midrule
\endfirsthead

\toprule
Figure and experiment &
Generator equation &
Numerical settings and compact-SIR procedure &
Completed result and interpretation\\
\midrule
\endhead

\bottomrule
\endfoot

\textbf{Figure 1}\par
Controlled one- through five-wave epidemics
&
\(\displaystyle
\begin{aligned}
\mu(t)
&=25+\sum_jH_j
  \exp\!\left[-\{(t-\psi_j)/s_j\}^2/2\right],\\[-1pt]
y_t&\sim\mathrm{NegBin}(\mu_t,\phi=80).
\end{aligned}
\)

\smallskip
The noiseless incidence target is mapped to a discrete SIR oracle.
&
\(\displaystyle
\begin{aligned}
\textstyle N=10^6, S_0=995000, I_0=800, \gamma=1/7.
\end{aligned}
\)

Candidate dimensions are $K=3,\ldots,22$. Oracle beta is projected to each
$K$-knot $\log\text{-}\beta$ basis and propagated by forward SIR simulation. For each
$K$, 120 independently generated negative-binomial realizations are assessed.
&
\(\displaystyle K^\star=8,8,11,14,18\) for one through five waves.
Selected-$K$ adequacy probabilities are $0.99,0.96,0.76,0.68,0.74$.
MPPC depends on wave geometry, noise, and basis placement, not peak count
alone.\\
\addlinespace

\midrule
\textbf{Figure 2}\par
SEIR compression
&
\(\displaystyle
\begin{aligned}
\dot S&=-\beta(t)SI/N,\\
\dot E&=\beta(t)SI/N-\sigma_EE,\\
\dot I&=\sigma_EE-\gamma I,\\
\dot R&=\gamma I,
\end{aligned}
\)

\(\displaystyle
\begin{aligned}
\beta(t)={}&0.055
 +0.31e^{-(t-38)^2/(2\cdot15^2)}\\
&+0.34e^{-(t-118)^2/(2\cdot17^2)}\\
&+0.37e^{-(t-196)^2/(2\cdot15^2)}.
\end{aligned}
\)
&
\(\displaystyle
(S_0,E_0,I_0,R_0)=(998400,800,800,0)
\)

\smallskip
\(\displaystyle \sigma_E=1/4, \gamma=1/7.\)

\smallskip
A deterministic incidence oracle is projected to a $\log\text{-}\beta$ knot basis and
forward-simulated with compact SIR dimension $K=14$.
&
The compact SIR fit preserves 3 of 3 major peaks; log-NRMSE is $0.362$.
The fit compresses SEIR incidence, but its effective beta is not the original
SEIR transmission pulse.\\
\addlinespace

\midrule
\textbf{Figure 3}\par
Age-structured aggregate compression
&
For $i=1,2,3$:
\[
\textstyle
\begin{aligned}
\dot S_i&=-S_i\sum_jg_i(t)C_{ij}I_j/N_j,\\
\dot I_i&=S_i\sum_jg_i(t)C_{ij}I_j/N_j-\gamma I_i,\\
\dot R_i&=\gamma I_i.
\end{aligned}
\]
&
Population shares are $(0.30,0.45,0.25)$, with
\[
\textstyle
C=
\begin{psmallmatrix}
3.4&1.2&0.5\\
0.9&3.0&0.8\\
0.4&1.0&2.3
\end{psmallmatrix}.
\]

Group phases are $(-9,0,11)$ days,
$I_0=(900,600,300)$, and $\gamma=1/7$. Only aggregate incidence is supplied
to deterministic oracle--knot--forward-SIR compression with $K=15$.
&
The aggregate SIR fit preserves 3 of 3 broad peaks. Aggregate wave topology
can be retained even though subgroup incidence, group-specific forcing, and
the contact matrix are not identified by the aggregate SIR fit.\\
\addlinespace

\midrule
\textbf{Figure 4}\par
Two-strain aggregate compression
&
\(\displaystyle
\begin{aligned}
\dot S&=-S(\beta_1I_1+\beta_2I_2)/N,\\
\dot I_1&=\beta_1SI_1/N-\gamma I_1,\\
\dot I_2&=\beta_2SI_2/N-\gamma I_2+\iota(t),\\
\dot R&=\gamma(I_1+I_2)-\iota(t).
\end{aligned}
\)

\(\displaystyle
\begin{aligned}
\beta_1(t)&=0.050+0.32e^{-(t-45)^2/(2\cdot13^2)},\\
\beta_2(t)&=0.050+0.38e^{-(t-160)^2/(2\cdot17^2)},\\
\iota(t)&=45e^{-(t-110)^2/(2\cdot3^2)}.
\end{aligned}
\)
&
The strains share the same susceptible population and use $\gamma=1/7$.
A deterministic incidence oracle is compressed into a single-strain SIR with
$K=14$. Frequency-weighted transmissibility is plotted only as a diagnostic.
&
The compact SIR fit preserves 2 of 2 major peaks; log-NRMSE is $0.192$.
The fitted beta is a composition-dependent effective forcing, not a
strain-specific biological transmission parameter.\\
\addlinespace

\midrule
\textbf{Figure 5}\par
Removal-rate sweep and exact non-identifiability
&
For the same fixed three-wave incidence:
\[
\textstyle
\begin{aligned}
S_{t+1}&=S_t-c_t,\\
I_{t+1}&=I_t+c_t-\gamma I_t,\\
\beta_t&=\frac{Nc_t}{S_tI_t}.
\end{aligned}
\]
&
\(\displaystyle
\gamma\in\left\{
\frac1{14},\frac1{10},\frac17,\frac15,\frac13
\right\}.
\)

\smallskip
No ABC or compact-knot fitting is used. The incidence trajectory is held fixed
and reconstructed exactly under each removal-rate assumption.
&
All five choices reproduce incidence exactly for infectious durations from 3
to 14 days. The $I_{\max}$ ratio is $3.90$ and the $\beta_{\max}$ ratio is
$2.51$. Incidence alone does not identify removal, prevalence, or transmission
separately.\\
\addlinespace

\midrule
\textbf{Figure 6}\par
Peak-topology robustness to reporting noise
&
\(\displaystyle
y_t\sim\mathrm{NegBin}(\mu_t,\phi),
\qquad
\phi\in\{\infty,100,30,10,3\}.
\)
&
There are 120 independent reporting realizations at each noise level. Peak
detection uses a 7-day smooth, 18-day minimum separation, and 7\% prominence
threshold. No SIR fit is performed.
&
Correct three-peak recovery is
$100\%,\ 100\%,\ 100\%, 94.2\%,\ 32.5\%$
from Poisson noise to $\phi=3$. Severe overdispersion produces false peaks and
unstable observed topology.\\
\addlinespace

\midrule
\textbf{Figure 7}\par
Philippine COVID-19 2020--2023 compact reconstruction
&
Reported daily cases from 30 January 2020 to 21 April 2023 are smoothed over
7 days to define the incidence oracle. A positive $\log\text{-}\beta$ grid is calibrated
by repeated forward-SIR simulation with a weak second-difference penalty.
&
\(\displaystyle
N=109\text{ million},\qquad
S_0=108999900,
\)

\smallskip
\(\displaystyle
I_0=10, \gamma=1/7.
\)

\smallskip
A data-span-derived broad-wave separation rule automatically detects five
broad waves and allocates local knot resolution; the number of waves is not
fixed manually. The compact reconstruction uses $K=36$.
&
There are 1,178 observations and 4,087,964 reported cases. Overall log-NRMSE
against raw reports is $0.185$; peak-window log-NRMSE against the 7-day target
is $0.156$. This is descriptive reconstruction, not daily-oracle fitting or
causal attribution.\\
\addlinespace

\midrule
\textbf{Figure 8}\par
Future-knot projection bridge from Philippine COVID-19 2020
&
The historical 2020 reconstruction supplies pre-boundary
$\widehat{\beta}(t)$. Future log-transmission knots generate
$\beta_{\mathrm{base}}(t)$.

\smallskip
The variant-informed scenario is
\[
\textstyle
\beta_{\mathrm{var}}(t)
=
\min\left\{
0.19,\,
\beta_{\mathrm{base}}(t)
[1+0.75r_{\mathrm{var}}(t)]
\right\}.
\]
&
Only 2020 Philippine COVID-19 data initialize the projection. There are 8
future knots over a 90-day horizon. The advance-information scenario specifies
a variant arriving 35 days after the boundary, a 1.75-fold transmission
multiplier, and a 7-day transition. The 2021 series is retained only as an
orientation curve.
&
Four panels compare no-advance-information and variant-informed scenarios:
incidence on the top row and beta on the bottom row. The variant scenario is
conditional on its declared arrival, multiplier, transition, and beta ceiling;
it is not a retrospective fit to 2021 or validated long-term forecasting.\\

\end{longtable}

\endgroup
\end{landscape}

For clarity, the distance used in the fitted examples is
\begin{equation}
D
=
d_{\mathrm{curve}}
+
4\left(
d_{\mathrm{count}}
+
d_{\mathrm{time}}
+
0.35d_{\mathrm{height}}
+
0.25d_{\mathrm{trough}}
\right)
+
0.5d_{\mathrm{cum}},
\label{eq:distance-used-in-experiments}
\end{equation}
which is the numerical diagnostic used to compare trajectories. Equivalently,
the component weights are
\[
w_c=1,\qquad
w_n=4,\qquad
w_t=4,\qquad
w_h=1.4,\qquad
w_v=1,\qquad
w_q=0.5.
\]
Thus, the whole-curve discrepancy has baseline weight \(1\), peak count and
peak timing receive weight \(4\), peak height receives the effective weight
\(4(0.35)=1.4\), trough depth receives the effective weight \(4(0.25)=1\),
and cumulative-incidence timing receives weight \(0.5\).

For the controlled MPPC scan, adequacy was evaluated over 120 independently
generated negative-binomial observation realizations for each \(K\). A
realization had to reproduce the detected peak count and satisfy log-scale
normalized RMSE \(\le 0.72\), peak-time MAE \(\le 7\) days, peak-height
MARE \(\le 0.35\), and trough-depth MARE \(\le 0.50\). A value of \(K\)
was classified as adequate when at least \(60\%\) of the realizations passed
these criteria. The other synthetic compression panels are single
deterministic forward simulations from their stated generators, and the
reporting-noise sweep uses 120 independent replicates per noise level.

The weights in \(D\) are not fixed laws. They are user-chosen priorities that
tell the fitting or screening procedure which kinds of mismatch should be
treated as more serious. The value \(w_c=1\) sets the whole-curve fit as the
baseline: the fitted model should still follow the overall epidemic
trajectory, not only match the detected peaks. The value \(w_n=4\) strongly
penalizes the wrong number of peaks. This is important because the goal of
the analysis is peak preservation: a fit with one peak should not be
considered adequate if the observed curve has three meaningful peaks, even
when the average curve error is acceptable.

The value \(w_t=4\) strongly penalizes wrong peak timing. Thus, a model is
not considered adequate merely because it reproduces the correct number of
waves; the waves must also occur at approximately the correct times. The
value \(w_h=1.4\) gives moderate emphasis to peak height. Peak height is
epidemiologically important, but it is weighted less heavily than peak count
and peak timing because it is often more sensitive to reporting noise,
surveillance changes, and stochastic fluctuation.

The value \(w_v=1\) penalizes mismatch in trough depth, but less strongly
than peak count and peak timing. This is reasonable because troughs help
distinguish genuinely separated waves from small undulations, but trough
estimates can be unstable when reported incidence is low. The value
\(w_q=0.5\) gives a smaller penalty to cumulative-incidence timing. This term
helps prevent fits that match the visible peaks but shift the overall
epidemic burden too early or too late.

The reason for this weighting scheme is that MPPC is not ordinary curve
fitting. It is peak-preserving compression. Therefore, the distance must
prioritize whether the fitted trajectory preserves the topology and geometry
of the observed epidemic waves.

\subsection{Experiment 1: controlled one- through five-peak incidence}
We generated one through five well-separated Gaussian incidence waves over 241 days and used the incidence representability theorem to map each noiseless target to a discrete SIR oracle for the selected $\gamma$. For each $K=3,\ldots,22$, the oracle beta was projected to a $K$-knot $\log\text{-}\beta$ basis and then propagated through the forward SIR state updates. Negative-binomial reporting noise with dispersion $\phi=80$ was generated afresh in 120 replicates per $K$ to assess adequacy. Figure~\ref{fig:controlled-peaks} summarizes this simulation-based scan. The resulting $K^\star$ values are design-conditional operational minima: they depend on the specified wave geometry, observation window and resolution, smoothing and peak-detection rules, reporting-noise model, knot basis, and adequacy thresholds. They should therefore not be interpreted as universal values determined solely by the number of epidemic peaks.

The selected fast-scan MPPC values were $K^*=8,8,11,14,$ and $18$ for one through five designed peaks, respectively. At the selected $K^\star$, the simulated adequacy probabilities were 0.99, 0.96, 0.76, 0.68, and 0.74. Importantly, $K^\star$ was not strictly monotone in the number of peaks: the one- and two-wave geometries required the same fitted dimension. This is a useful refinement of the MPPC concept: dynamic complexity depends on peak width, separation, relative height, trough depth, observation noise, and basis placement, not merely on the integer number of waves.

\begin{figure}[htbp]
\centering
\includegraphics[width=0.88\textwidth]{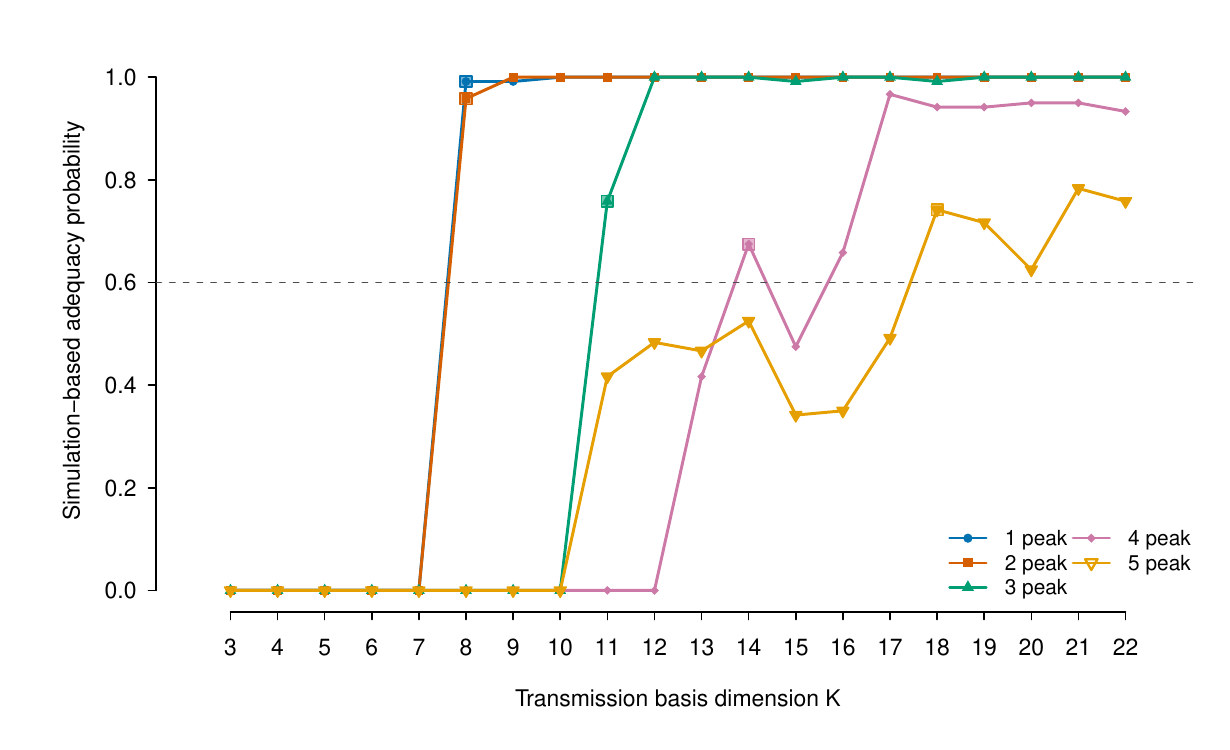}
\caption{Minimal peak-preserving complexity across one- through five-wave epidemics. Curves show the simulated adequacy probability across nested transmission dimensions; the dashed line is the 0.60 threshold and square markers identify the first adequate $K$ when one exists. MPPC depends on wave geometry and noise rather than peak count alone.}
\label{fig:controlled-peaks}
\end{figure}

For Experiments~\ref{fig:seir-compression}--\ref{fig:strain-compression}, each asks a different question: can SIR hide an exposed stage, can it hide age-group heterogeneity, and can it hide variant turnover? A similar case curve can still come from different hidden stories.

\subsection{Experiment 2: SEIR forcing compressed into SIR}
We generated a three-wave SEIR epidemic with mean latent duration 4 days, mean infectious duration 7 days, and three smooth transmission pulses. The compact SIR omitted the exposed compartment. With $K=14$, its deterministic forward trajectory preserved all three detected peaks (log-scale normalized RMSE 0.362). The effective transmission history is therefore a trajectory-compression device rather than the SEIR transmission parameter: omitted latency changes the relation between the generator pulse and the single-SIR forcing (Figure~\ref{fig:seir-compression}).

\begin{figure}[htbp]
\centering
\includegraphics[width=0.84\textwidth]{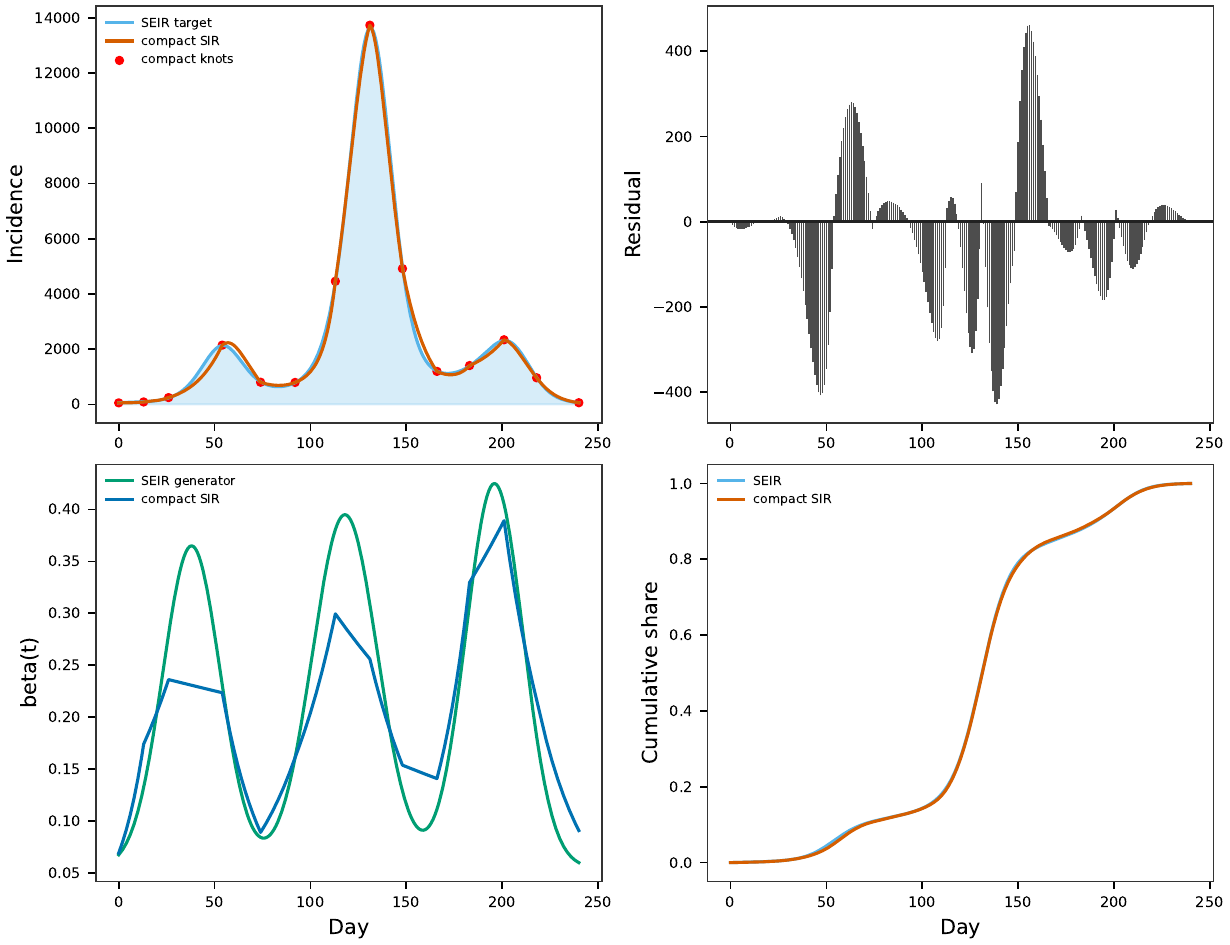}
\caption{SEIR-to-SIR compression. The shaded incidence panel compares the latent-stage target with the SIR compression; the residual panel shows where the simple model misses timing or amplitude; the beta panel compares generator transmission with fitted effective $\beta(t)$; and the cumulative panel checks whether total epidemic timing is preserved. Red dots mark the compact SIR knots. The recovered $\beta(t)$ is a shifted effective forcing, not the original SEIR transmission pulse.}
\label{fig:seir-compression}
\end{figure}

\subsection{Experiment 3: age-mixing heterogeneity}
We generated incidence from a three-group SIR system with a fixed non-diagonal contact matrix and group-specific phase shifts in three forcing episodes, reflecting the well-established importance of age-assortative social mixing in transmission models \citep{mossong2008,prem2017}. Only total incidence was supplied to the compact SIR reconstruction. With $K=15$, peak-aware knot placement followed by forward-SIR calibration improves the aggregate trajectory while preserving all three automatically detected peaks. This good aggregate match must not be read as recovery of subgroup incidence or the contact matrix: neither enters the reconstruction target. The experiment therefore identifies a boundary of dynamic parsimony: aggregate wave topology can survive compression even when subgroup dynamics are not recoverable (Figure~\ref{fig:age-compression}).

\begin{figure}[htbp]
\centering
\includegraphics[width=0.84\textwidth]{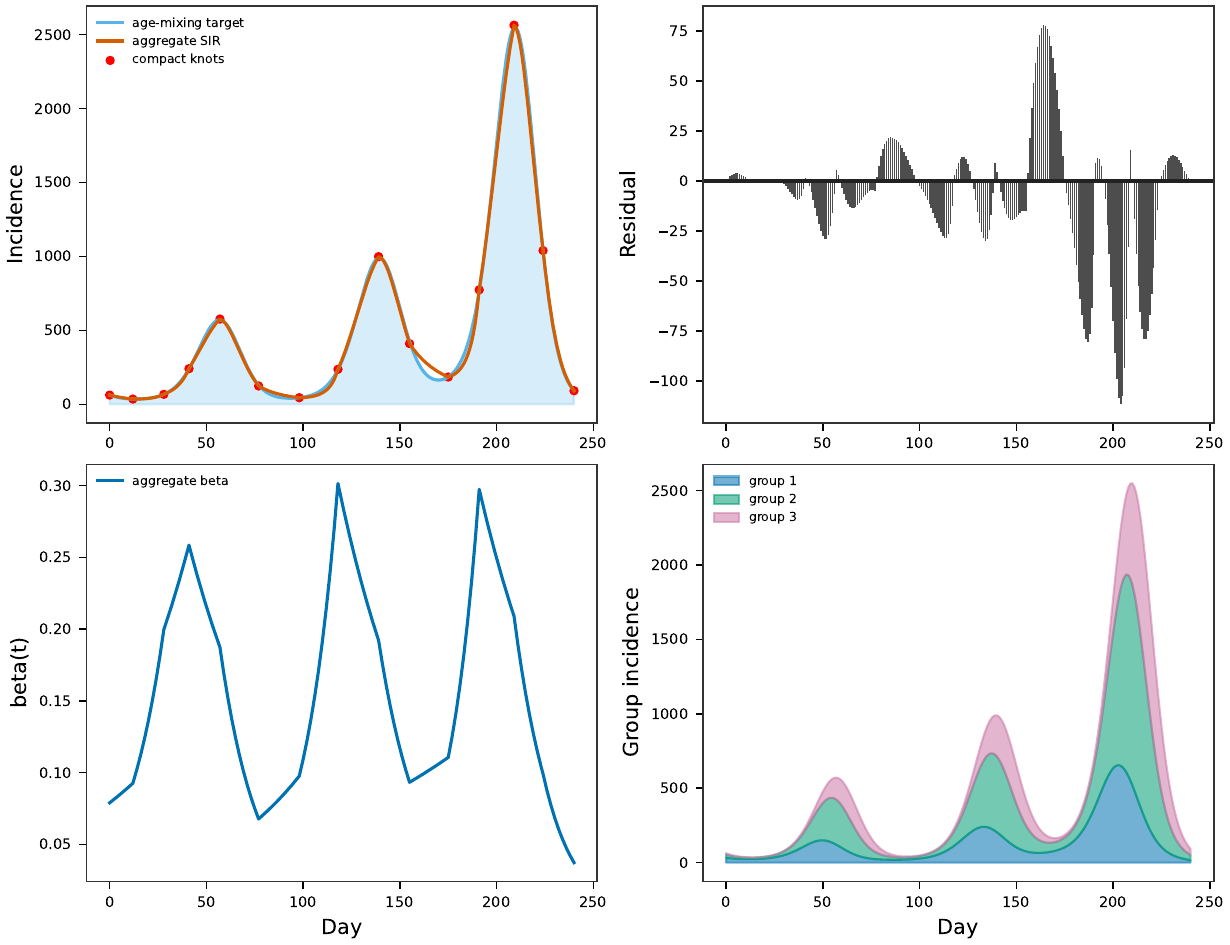}
\caption{Age-structured-to-aggregate SIR compression. The incidence and residual panels evaluate aggregate fit, the beta panel shows the single fitted effective $\beta(t)$, and the group-composition panel shows the hidden age-specific incidence that the aggregate SIR did not observe. Red dots mark the compact SIR knots. The aggregate curve can be acceptable even when subgroup dynamics are not identified.}
\label{fig:age-compression}
\end{figure}

\subsection{Experiment 4: two-strain replacement}
We simulated two competing strains sharing the same susceptible pool. Strain 1 generated an early wave; strain 2 was introduced later and had a distinct higher transmission pulse. A single-strain SIR with $K=14$ preserved both aggregate peaks in its deterministic forward reconstruction (log-scale normalized RMSE 0.192). This illustrates how a simple effective forcing can absorb variant turnover while simultaneously warning that the resulting $\beta(t)$ is a composition-dependent lumped quantity, not a strain-specific biological parameter (Figure~\ref{fig:strain-compression}).

\begin{figure}[htbp]
\centering
\includegraphics[width=0.84\textwidth]{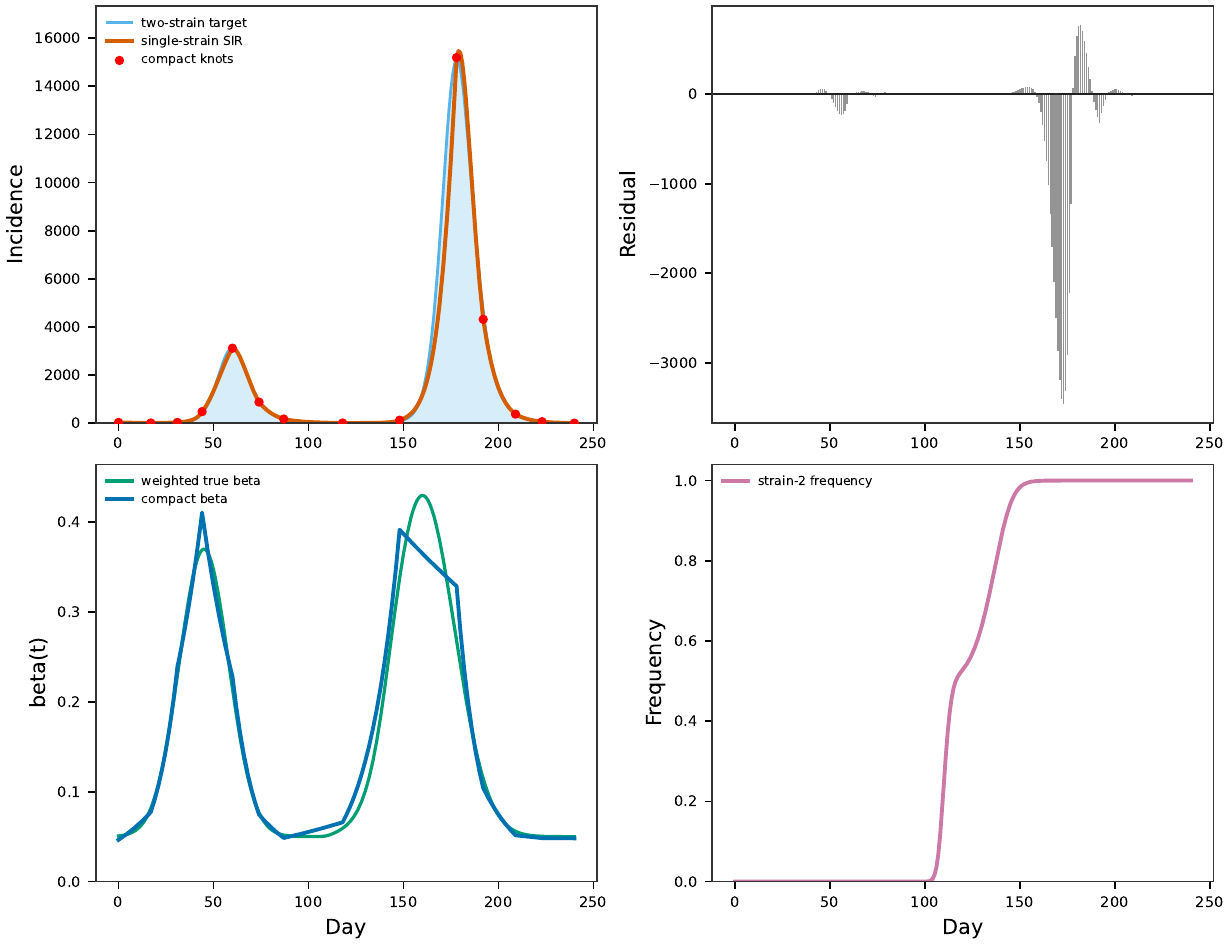}
\caption{Two-strain-to-single-SIR compression. The incidence and residual panels evaluate the aggregate single-SIR fit, the beta panel compares fitted $\beta(t)$ with the frequency-weighted transmissibility, and the final panel shows strain-2 frequency. Red dots mark the compact SIR knots. The fitted $\beta(t)$ is therefore a composition-dependent effective quantity, not a strain-specific biological parameter.}
\label{fig:strain-compression}
\end{figure}

\subsection{Experiment 5: removal-rate sweep and exact non-identifiability}
To demonstrate Theorem~\ref{thm:gamma-nonident} directly, we fixed the same smooth three-peak incidence curve and reconstructed the exact discrete SIR oracle under $\gamma\in\{1/14,1/10,1/7,1/5,1/3\}$ day$^{-1}$. Incidence RMSE is exactly zero for every reconstruction by construction. Nevertheless, peak infectious prevalence changes from approximately 50,068 at a 14-day infectious duration to 12,831 at a 3-day duration, a 3.90-fold range. Maximum reconstructed $\beta(t)$ changes from 0.237 to 0.595 day$^{-1}$, a 2.51-fold range, and the corresponding $\Reff(t)$ trajectories also differ. Therefore incidence fit alone cannot determine removal, infectious prevalence, and effective transmission separately (Figure~\ref{fig:gamma-sweep}).

\begin{figure}[htbp]
\centering
\includegraphics[width=0.84\textwidth]{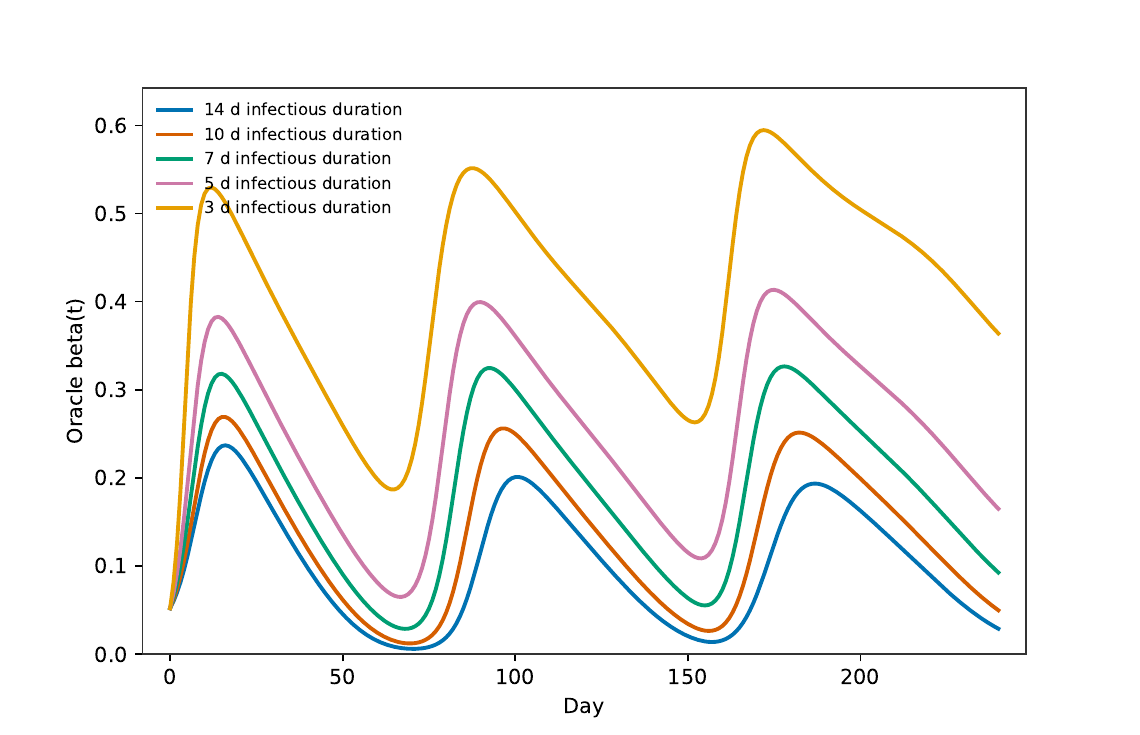}
\caption{Removal-rate non-identifiability under incidence-equivalent oracle reconstruction. A fixed controlled three-peak incidence curve is represented exactly for five removal rates. The observed incidence is unchanged with a fitting error of $\text{RMSE}=0$, but the implied infectious prevalence, effective transmission history $\beta(t)$, and effective reproduction trajectory $\mathcal R_{\mathrm{e}}(t)$ differ substantially. Legend: Each line corresponds to a different infectious duration, $1/\gamma$.}
\label{fig:gamma-sweep}
\end{figure}

If the parameters were identifiable, only one specific value of $\gamma$ would successfully fit the incidence curve, and all other $\gamma$ values would fail or produce massive fit errors. The fact that any plausible $\gamma$ works (by forcing $\beta(t)$ to compensate) implies structural non-identifiability.

The practical message of Figure~\ref{fig:gamma-sweep} is that the same reported case curve can hide very different underlying epidemics. A short infectious duration and a long infectious duration can both reproduce the same incidence, but they imply different numbers of infectious people and different transmission histories. Thus, a good case-curve fit alone is not enough to claim that the model has identified the true infectious period, hidden prevalence, or causal transmission rate.

\subsection{Experiment 6: negative-binomial reporting noise}
We quantified how reporting overdispersion affects the peak topology supplied to the fitting algorithm. For each of five noise levels we generated 120 independent realizations of the same three-peak mean curve and applied the fixed 7-day, 18-day-separation, 7\%-prominence detector. Correct peak-count recovery was 100\% under Poisson noise, 100\% at $\phi=100$, 100\% at $\phi=30$, 94.2\% at $\phi=10$, and only 32.5\% at $\phi=3$. At $\phi=3$, the detector generated an average of 0.89 extra peaks per series and mean peak-time MAE was 11.3 days. The practical implication is important: peak-aware fitting cannot rescue a peak topology that is itself unstable to the observation process. Peak-detection sensitivity analysis is therefore part of model validation rather than a cosmetic preprocessing choice (Figure~\ref{fig:noise-sweep}).

\begin{figure}[htbp]
\centering
\includegraphics[width=0.84\textwidth]{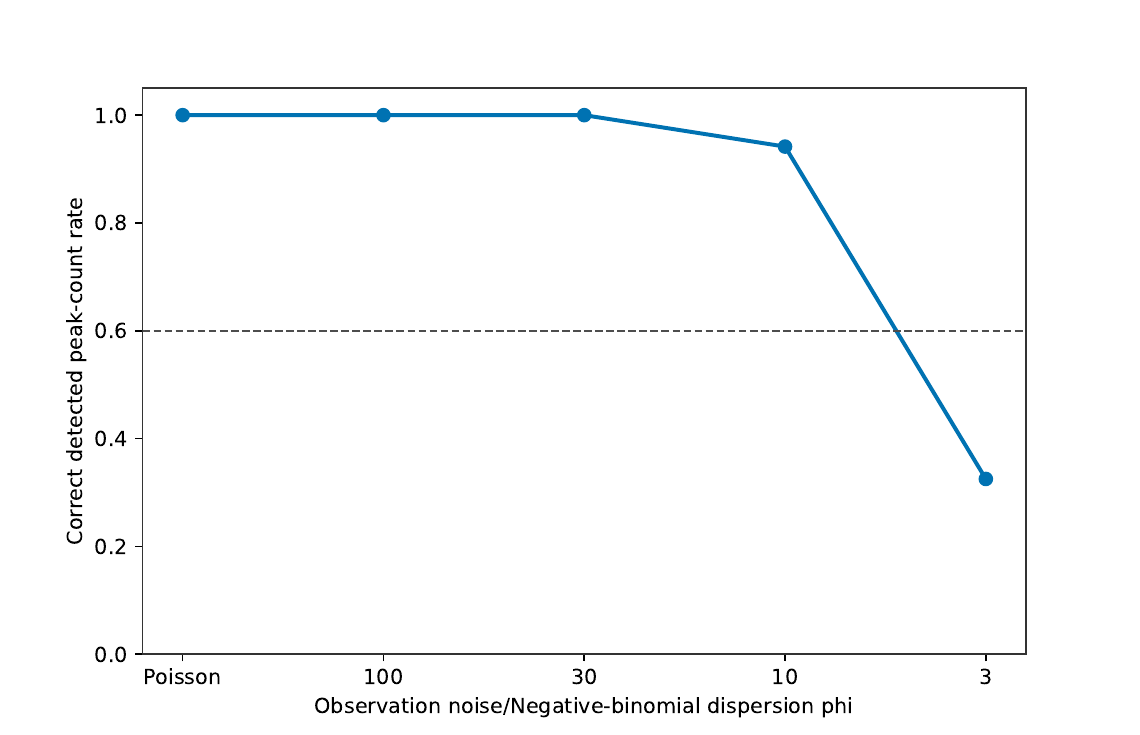}
\caption{Peak-topology robustness to reporting noise. Strong negative-binomial overdispersion can create spurious detected waves even when the latent mean trajectory has exactly three peaks.}
\label{fig:noise-sweep}
\end{figure}

\subsection{Experiment 7: Philippine COVID-19 reported-case proof-of-concept}
We use the Philippine COVID-19 reporting-date series as an empirical proof-of-concept, complementing descriptive analytics work that summarizes the Philippine pandemic trajectory using epidemiological, mobility, policy, and timeline information \citep{bansilan2024descriptive}. The data contain reported cases from 30 January 2020 to 21 April 2023. The daily series contains 1,178 observations and 4,087,964 reported cases. A 7-day smooth and data-span-derived minimum separation automatically identify five broad waves; the rule suppresses reporting sub-peaks without setting the number of waves by hand.

For Figure~\ref{fig:philippines-covid}, the reconstructed epidemic curve uses a compact $K=36$ $\log\text{-}\beta$ basis over 1,178 daily observations. Knot locations are selected deterministically from the five automatically detected broad waves and then the knot values are calibrated by repeated forward SIR simulation, with a weak curvature penalty to discourage one-day reporting artifacts. The curve is intentionally descriptive rather than daily-oracle-perfect; the script writes its reconstruction error and beta range with each run. This illustrates the MPPC principle: use enough knots to preserve important waves without treating every reported fluctuation as a mechanism.

\begin{figure}[htbp]
\centering
\includegraphics[width=0.88\textwidth]{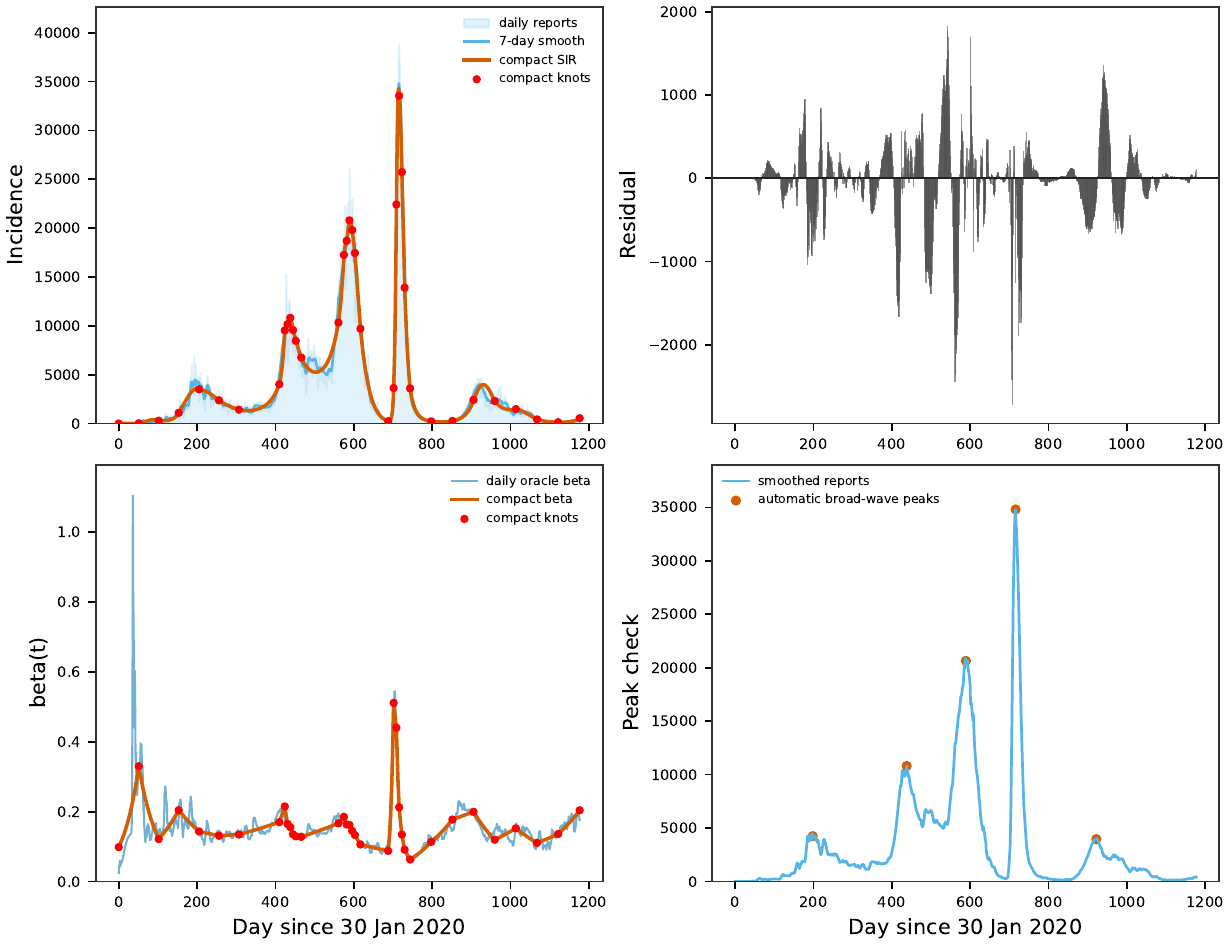}
\caption{Philippine COVID-19 reported cases and compact peak-preserving forward-SIR reconstruction for 2020--2023. The incidence panel shades daily reports and overlays the 7-day smooth and calibrated 36-knot forward simulation; the residual panel shows the remaining reconstruction error; the beta panel shows the daily oracle proxy and compact $\widehat\beta(t)$; and the peak-check panel marks five automatically detected broad waves. Red dots mark the compact SIR knots in the incidence and beta panels. The fitted $\beta(t)$ is an effective reporting-date forcing, not a direct causal estimate of contact behavior.}
\label{fig:philippines-covid}
\end{figure}

The Philippine COVID-19 example shows how the method behaves on real surveillance data rather than only simulated curves. Figure~\ref{fig:philippines-covid} uses a compact peak-preserving reconstructed curve, so the curve follows the major reported peaks without pretending to reproduce every daily reporting fluctuation. The message is descriptive: it summarizes the shape of reported cases and the effective beta needed to reproduce that shape, but it does not by itself identify which policies, behaviors, mobility trend, variants, testing changes, vaccination effects, local outbreaks, or reporting delays caused each movement.

\subsection{Experiment 8: modelling future knot values for short-term projection using Philippine COVID-19 2020 data}
The fitted retrospective curve is useful for description and short-term forecasting, but it is not automatically valid for long-term or ``what-if'' forecasts. The reason is the {lumped-driver problem}: $\beta(t)$ compresses many hidden causes into one function. If those causes change, extrapolating the last fitted curve can fail. It also creates {forward sensitivity to non-identifiability}: two retrospective fits with nearly identical incidence can imply very different future paths once their hidden assumptions are projected forward.

To bridge the gap without pretending that the simple model has identified all mechanisms, the simulation can add a future-knot layer. Let $\theta_j=\log\beta(\tau_j)$ be the fitted retrospective knot values. A minimal stochastic projection places future log-transmission knots after the observed window and evolves them as
\begin{equation}
\theta^+_{r+1}=\theta^+_r+a+\sigma_\theta\eta_r,
\qquad \eta_r\sim N(0,1),
\label{eq:future-rw-knot}
\end{equation}
where $a$ is estimated from recent retrospective knot drift and $\sigma_\theta$ is the empirical recent knot variability. This construction can generate a distribution of future $\beta(t)$ curves and therefore a distribution of future incidence trajectories. Figure~\ref{fig:future-knots} compares this no-advance-information continuation with a conditional advance-information scenario: an emerging variant is specified to arrive 35 days after the 2020 boundary, with a 1.75-fold transmission multiplier that ramps in over 7 days. If $r_{\mathrm{var}}(t)$ is the logistic replacement fraction, the latter uses $\beta_{\mathrm{var}}(t)=\min\{0.19,\beta(t)[1+0.75r_{\mathrm{var}}(t)]\}$. The 0.19/day ceiling is a declared short-horizon scenario bound, used to prevent implausibly explosive upper-tail paths. The 2021 reports are displayed only as an orientation curve and are not used to tune the multiplier. When external predictors are available, a more interpretable bridge is
\begin{equation}
\theta^+_{r}=\alpha_0+x_r^{+\top}\alpha+u_r,
\qquad
u_{r+1}=\varphi u_r+\sigma_u\eta_r,
\label{eq:future-covariate-knot}
\end{equation}
where $x_r^+$ may include projected mobility, school-calendar indicators, testing intensity, variant frequencies, rainfall, vaccination coverage, or policy indicators. The exogenous covariate model is preferable for scenario analysis because it states which future driver is being changed.

\begin{figure}[htbp]
\centering
\includegraphics[width=0.88\textwidth]{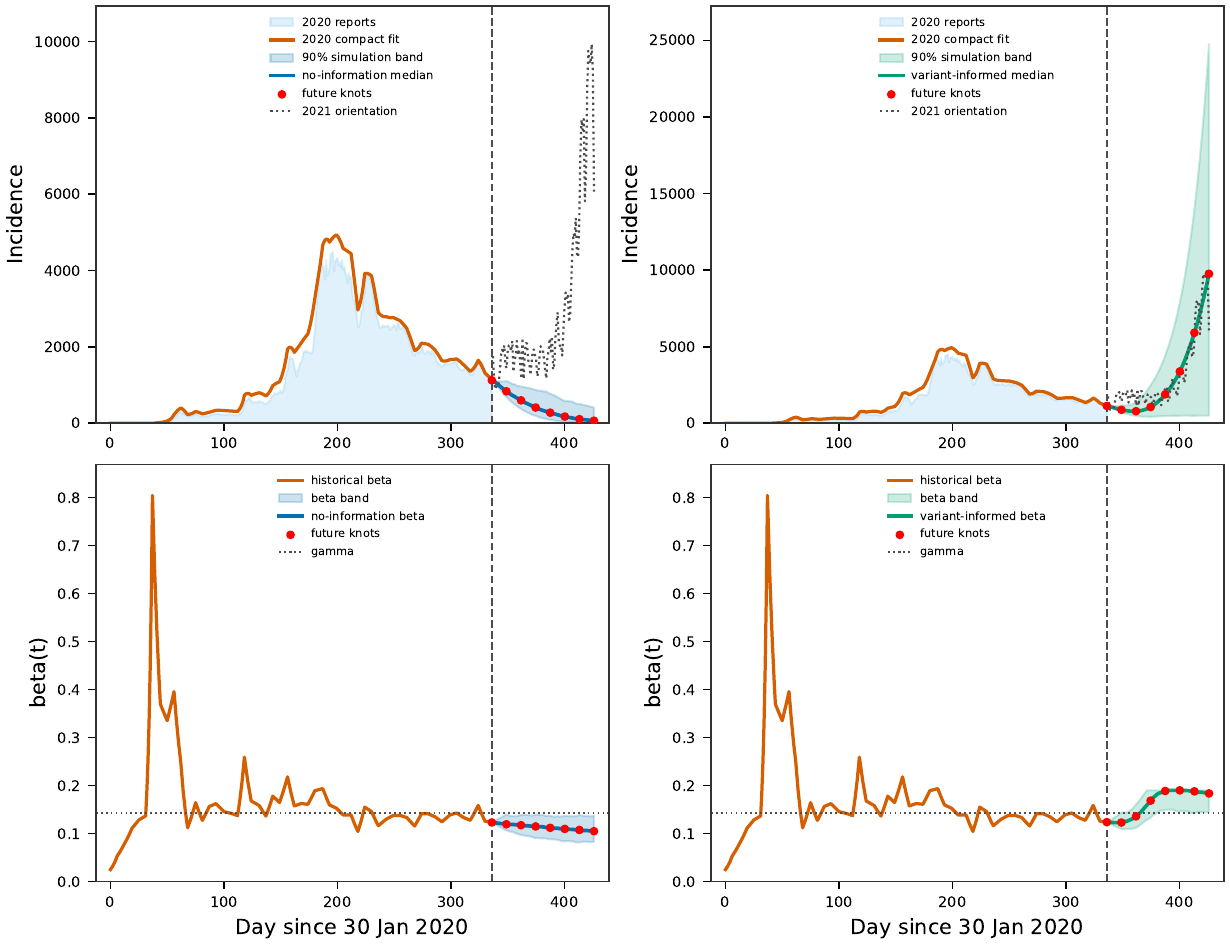}
\caption{Short-term future-knot projections initialized from Philippine COVID-19 2020 data. The top row shows incidence and the bottom row the corresponding beta trajectories. The left column is the no-advance-information continuation of the declining 2020 state; the right column is conditional on advance knowledge of a variant arriving 35 days after the boundary with a 1.75-fold transmission multiplier and a declared 0.19/day beta ceiling. Red dots mark the compact future knots. In both incidence panels, the dotted 2021 series is an orientation curve only; it is not used to estimate the scenario parameters. The figure illustrates why explicit advance information can change a forecast qualitatively, not validated prospective or long-term forecasting skill.}
\label{fig:future-knots}
\end{figure}

For forecasting, the retrospective fitted curve should not be treated as the data-generating truth. Instead, several plausible future bend-points for $\beta(t)$ can be specified and propagated to show the range of epidemic paths they imply. For policy ``what-if'' questions, those future bend-points should be linked to explicit drivers such as mobility, testing, variants, vaccination, or school opening rather than extrapolated only from the past.

\subsection{Cross-experiment interpretation}

For every experiment with an estimated or reconstructed time-varying transmission function, the reported curve can be written in the common form
\begin{equation}
\widehat\beta_e(t)=\exp\left\{(1-w_{e,j}(t))\widehat\theta_{e,j}+w_{e,j}(t)\widehat\theta_{e,j+1}\right\},
\qquad t\in[\tau_{e,j},\tau_{e,j+1}],
\label{eq:experiment-beta-output}
\end{equation}
where $e$ indexes the experiment and $w_{e,j}(t)=(t-\tau_{e,j})/(\tau_{e,j+1}-\tau_{e,j})$. Table~\ref{tab:beta-profiles} gives the resulting fitted or oracle-projected $\beta(t)$ summaries. The Philippine empirical row corresponds to the full 2020--2023 case series and the reconstructed epidemic curve shown in Figure~\ref{fig:philippines-covid}; the future-knot row is deliberately based only on Philippine COVID-19 data from 2020, as shown in Figure~\ref{fig:future-knots}.

\begin{table}[H]
\centering
\caption{Resulting time-varying transmission summaries by experiment. Values are per day. Controlled experiments and the removal-rate sweep use the exact incidence-oracle curve; compression and empirical experiments use fitted compact effective transmission; the future-knot row uses the median projected trajectory.}
\label{tab:beta-profiles}

\scriptsize
\setlength{\tabcolsep}{3pt}
\renewcommand{\arraystretch}{1.08}

\begin{tabularx}{\textwidth}{@{}
>{\raggedright\arraybackslash}X
>{\centering\arraybackslash}p{1.55cm}
>{\centering\arraybackslash}p{1.50cm}
>{\centering\arraybackslash}p{1.60cm}
>{\centering\arraybackslash}p{1.50cm}@{}}
\toprule
Experiment &
Complexity &
\shortstack{Minimum\\$\beta_{\mathrm{eff}}(t)$} &
\shortstack{Median\\$\beta_{\mathrm{eff}}(t)$} &
\shortstack{Maximum\\$\beta_{\mathrm{eff}}(t)$}\\
\midrule
Controlled 1 peak & $K^*=8$  & 0.027 & 0.144 & 0.338\\
Controlled 2 peaks & $K^*=8$ & 0.027 & 0.156 & 0.318\\
Controlled 3 peaks & $K^*=11$ & 0.029 & 0.185 & 0.327\\
Controlled 4 peaks & $K^*=14$ & 0.052 & 0.204 & 0.318\\
Controlled 5 peaks & $K^*=18$ & 0.052 & 0.223 & 0.344\\
SEIR compression & $K=14$ & 0.068 & 0.193 & 0.358\\
Age-structured aggregate compression & $K=15$ & 0.038 & 0.143 & 0.259\\
Two-strain aggregate compression & $K=14$ & 0.050 & 0.102 & 0.409\\
Removal-rate sweep & -- & 0.006 & 0.195 & 0.595\\
Philippine COVID-19 2020--2023 compact reconstruction & $K=36$ & 0.063 & 0.143 & 0.512\\
Future-knot projection from Philippine COVID-19 2020 & $K=8$ & 0.105 & 0.114 & 0.123\\
\bottomrule
\end{tabularx}

\vspace{2pt}
\parbox{\textwidth}{\scriptsize
\textit{Fit quality for the Philippine COVID-19 \(K=36\) reconstruction.}
The overall log-scale normalized RMSE against raw daily reports is 0.185;
the peak-window log-scale normalized RMSE against the 7-day-smoothed target
is 0.156. These are descriptive reconstruction errors, not prospective
forecast errors.}
\end{table}

Taken together, the eight experiments support four conclusions. First, multiple peaks are compatible with a structurally simple SIR whenever sufficiently flexible effective forcing is admitted. Second, the amount of required forcing is a property of wave geometry, resolution, and noise, not just wave count. Third, successful aggregate compression is deliberately weaker than mechanism identification: omitted latent stages, age mixing, and strain turnover can all be absorbed to different degrees into $\beta(t)$, while the $\gamma$ sweep proves that even perfect incidence reconstruction need not identify transmission and removal separately. Fourth, the Philippine COVID-19 and future-knot examples show how the same logic moves from synthetic validation to real reported-case reconstruction and short-horizon projection, provided the analysis remains explicit about lumped drivers and non-transportability. These results motivate reporting both trajectory adequacy and mechanism-relevant external validation.

Across the examples, the simple model is good at summarizing wave shapes, but different hidden causes or generators can produce similar-looking curves with  different scientific meanings. Therefore the fitted curve is useful, but it should not be over-interpreted without external evidence.

\section{Exploratory deployment for other diseases}
The framework is designed for public multi-wave surveillance series once the observation process and analysis window are stated explicitly. Suitable applications include COVID-19, influenza, dengue, measles, respiratory syncytial virus, and other infections with sufficiently resolved incidence data. A real-data analysis should report the original data source, reporting-date or onset-date limitations, revisions and missingness, smoothing used only for peak detection, assumptions about the population at risk, and the fact that reported cases generally undercount infections. The objective should be declared in advance: trajectory compression and descriptive forcing require less mechanistic information than causal attribution or intervention counterfactuals.

The Philippine COVID-19 example uses the 2020--2023 reporting-date series for descriptive reconstruction (Figure~\ref{fig:philippines-covid}), while the projection bridge deliberately uses only the 2020 portion for initialization (Figure~\ref{fig:future-knots}). This separation prevents the real-data proof-of-concept from being confused with a prospective validation exercise. The empirical claim is limited to dynamic compression unless external covariates or additional epidemiological observations are incorporated.

For real data, the method should be read as a disciplined summary of epidemic shape. It can say what kind of time-varying transmission curve would reproduce the reported waves, but it cannot by itself prove why the waves happened.

\subsection{Companion interactive application}
The accompanying Shiny app implements the full conceptual workflow, following the general framework in Figure \ref{fig:workflow}):
\begin{enumerate}
    \item upload a two-column daily incidence file or use built-in examples, including Philippine COVID-19 and the manuscript simulation generators;
    \item smooth only for peak detection and choose prominence/minimum separation;
    \item calculate the discrete incidence oracle for a chosen $\gamma$, $S_0$, and $I_0$;
    \item fit a positive piecewise-linear log-$\beta(t)$ representation using the peak-aware adaptive approximation;
    \item display posterior-like predictive incidence, residuals, effective transmission, and uncertainty summaries;
    \item compare nested numbers of log-$\beta$ knots to approximate MPPC;
    \item run the Simulation Lab examples with their equations, parameter values, and Eq.~\eqref{eq:distance-used-in-experiments} weights shown to the user;
    \item export posterior particles, fitted trajectories, reconstructed epidemic curves, beta summaries, figures, and tables.
\end{enumerate}
The Shiny app uses piecewise-linear basis in $\log\beta$ for speed and transparency. Other formulation can use generic smooth bases, such as cubic B-splines, with the same MPPC logic. The app retains upload support for user-supplied datasets while also providing the reproducible built-in examples used by the manuscript.

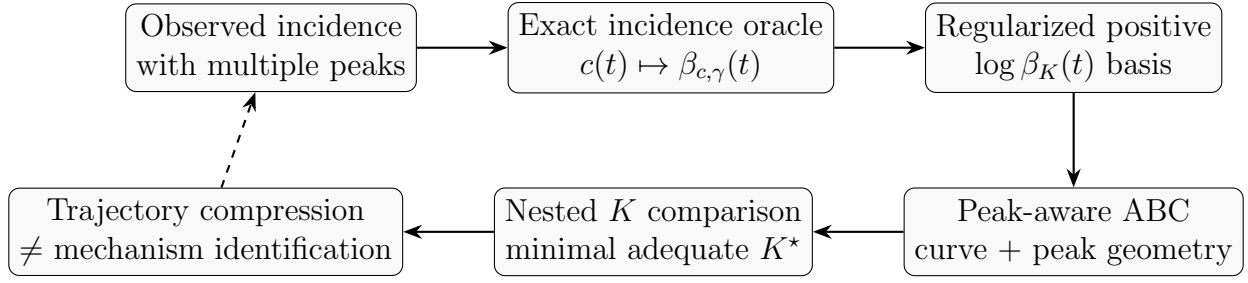
\begin{figure}[htbp]
\centering
\begin{tikzpicture}[
box/.style={draw,rounded corners,minimum width=3.25cm,minimum height=1.15cm,align=center,fill=gray!5},
arr/.style={-{Stealth[length=2.5mm]},thick}
]
\node[box] (data) {Observed incidence\\with multiple peaks};
\node[box,right=1.2cm of data] (oracle) {Exact incidence oracle\\$c(t)\mapsto\beta_{c,\gamma}(t)$};
\node[box,right=1.2cm of oracle] (reg) {Regularized positive\\$\log\beta_K(t)$ basis};
\node[box,below=1.25cm of reg] (abc) {Peak-aware ABC\\curve + peak geometry};
\node[box,left=1.2cm of abc] (mppc) {Nested $K$ comparison\\minimal adequate $K^\star$};
\node[box,left=1.2cm of mppc] (interp) {Trajectory compression\\$\neq$ mechanism identification};
\draw[arr] (data) -- (oracle);
\draw[arr] (oracle) -- (reg);
\draw[arr] (reg) -- (abc);
\draw[arr] (abc) -- (mppc);
\draw[arr] (mppc) -- (interp);
\draw[arr,dashed] (interp) -- (data);
\end{tikzpicture}
\caption{Peak-preserving dynamic parsimony workflow. The exact oracle establishes representability; regularization and MPPC determine how much time variation is minimally needed; inference is explicitly separated from mechanistic interpretation.}
\label{fig:workflow}
\end{figure}

The app is a guided laboratory: start with an epidemic curve, ask how much changing transmission is needed to reproduce the important waves, then check whether the result is only a good summary or whether additional data are needed for mechanism.

\subsection{Discrete incidence oracle used in the interactive application}
For a daily series, the app implements a discrete accounting analogue of Theorem~\ref{thm:incidence-representability}. Let $c_t$ denote the incidence assigned to day $t$ (raw or smoothed according to the selected oracle setting). The derivation is:
\begin{enumerate}[label={Step \arabic*.},leftmargin=*]
\item{Susceptible update.} Incidence removes $c_t$ individuals from $S$ during the interval, so
\[
S_{t+1}=S_t-c_t.
\]
The implementation caps $c_t$ below $S_t$ to prevent a negative susceptible state when a user supplies an incompatible population scale.

\item{Infectious update.} New infections enter $I$ and a fraction $\gamma$ leaves during a one-day Euler accounting step:
\[
I_{t+1}=I_t+c_t-\gamma I_t.
\]
For a general interval $\Delta t$, the corresponding first-order update is $I_{t+1}=I_t+c_t-\gamma I_t\Delta t$ if $c_t$ is an interval count; the implementation used here is designed for complete daily incidence series.

\item{Transmission reconstruction.} The discrete incidence identity is
\[
c_t=\beta_t\frac{S_tI_t}{N}.
\]
Solving for the effective transmission rate gives
\begin{equation}
\beta_t=\frac{Nc_t}{S_tI_t}.
\label{eq:discrete-oracle}
\end{equation}

\item{Effective reproduction number.} With externally chosen $\gamma$,
\[
\mathcal R_{\mathrm e,t}=\frac{\beta_tS_t}{\gamma N}.
\]
Changing $\gamma$ changes both the reconstructed infectious stock and the inferred $\beta_t$, even though the same incidence is preserved, exactly as demonstrated by Theorem~\ref{thm:gamma-nonident} and Experiment 5.
\end{enumerate}
The oracle is therefore an accounting benchmark rather than a causal estimator. Smoothing is used only when requested, and the raw incidence series remains available for comparison.

\section{Discussion}
This work reframes multi-wave epidemic modeling around a distinction between {structural state complexity} and {dynamic parameter complexity}. A model can remain structurally simple---only $S$, $I$, and $R$---while becoming dynamically expressive through a time-varying transmission function. Theorem~\ref{thm:incidence-representability} shows that, at the trajectory level, multiple peaks do not force the introduction of additional compartments. Theorem~\ref{thm:noresurgence} then identifies the dynamic requirement behind resurgences: transmission must increase sufficiently rapidly to overcome susceptible depletion in the effective reproduction number. Proposition~\ref{prop:incpeak} adds a separate condition for incidence peaks, which depend on both $\beta$ and $\dot\beta$.

This flexibility is simultaneously the main strength and main danger of the approach. With enough degrees of freedom, a time-varying transmission model can reproduce extremely broad classes of trajectories, as emphasized by inverse-problem studies \citep{pollicott2012,mummert2013}; more recent work shows that qualitative inverse-transmission patterns can be robust across model structures \citep{wang2026inverse}. The methodological gap addressed here is different: inverse reconstruction answers what forcing can reproduce the data, whereas MPPC asks how much forcing complexity is {minimally necessary} to preserve a prespecified epidemic topology. By requiring posterior-predictive adequacy for peaks, troughs, timing, and the complete trajectory across nested functional dimensions, MPPC turns flexibility into an auditable model-reduction quantity rather than an unrestricted curve-fitting privilege.

Peak-aware ABC is especially useful when the observed feature of interest is geometrical. Standard pointwise least squares can underweight a narrow but epidemiologically important resurgence, especially when a large first wave dominates the scale. Peak timing and magnitude are not new ABC summaries \citep{meyer2023}; what is new here is their role, together with peak count and trough structure, in a model-adequacy rule across nested transmission complexities. By including peak topology in the discrepancy and then defining MPPC as the first complexity level with adequate posterior-predictive peak preservation, the inference target is aligned with the scientific question while retaining an explicit parsimony principle.

The proposed framework also provides a rigorous warning about removal-rate inference. If only incidence is observed and $\beta(t)$ is unrestricted, every plausible removal rate can be paired with a different transmission function that gives the same incidence. Therefore, fitting a multi-peak incidence curve does not identify infectious duration, hidden prevalence, or causal transmission changes without additional data. This is exactly the layman message of Figure~\ref{fig:gamma-sweep}: the same reported curve can hide different underlying epidemics.

The empirical and projection additions extend the framework from synthetic validation to practical use. Figure~\ref{fig:philippines-covid} demonstrates compact reconstruction of the full Philippine COVID-19 2020--2023 series without overfitting every daily artifact. Figure~\ref{fig:future-knots} demonstrates how retrospective compression can be bridged to a short-term projection by modelling future knots, while making clear that such a bridge is not a validated long-term forecast. For policy scenarios, future knots should ideally be modeled using exogenous covariates, as in Eq.~\eqref{eq:future-covariate-knot}, rather than projected only from past beta drift.

MPPC is not a new likelihood, spline basis, or general replacement for AIC, BIC, cross-validation, Gaussian-process models, change-point models, or standard posterior-predictive checking. Rather, it is a peak-aware posterior-predictive parsimony criterion: among candidate time-varying transmission models, it selects the smallest complexity whose replicated incidence trajectories satisfy prespecified joint requirements for whole-curve agreement and preservation of peak count, timing, height, and trough structure. Standard criteria target related but different objectives—such as average pointwise fit, predictive density, held-out forecasting accuracy, or generic smoothness—and may therefore select a model that has favorable RMSE, AIC/BIC, or cross-validation performance while merging, shifting, or creating epidemiologically important peaks. Conversely, MPPC does not claim universal superiority: its result depends on the declared feature definitions, tolerances, observation resolution, smoothing, and transmission basis. Its specific contribution is to make the scientific objective explicit when the question is not simply “which model fits or forecasts best?” but “what is the least dynamic flexibility required for a simple epidemic model to reproduce the broad epidemic-wave structure judged important for this analysis?”

The main contribution is not that a simple model can be forced to fit many waves. The contribution is a disciplined way to ask how simple the changing-transmission curve can be while still preserving the important wave pattern, and when that fitted curve should or should not be interpreted as mechanism.

\subsection{Interpretation: what does the fitted effective transmission mean?}
A time-varying effective transmission rate is a {lumped driver}. Depending on the omitted mechanisms, variation in $\beta(t)$ may encode contact intensity and mixing structure; interventions or voluntary behavior; seasonal environmental forcing; variant replacement or changing infectiousness; spatial importation and synchronization; depletion or replenishment not represented by the simple susceptible compartment; vaccination, immunity, testing, or ascertainment changes if the observation model is misspecified.

Consequently, counterfactual policy interpretation requires more than historical fit. If an intervention changes a hidden mechanism differently from the way that mechanism was compressed into historical $\beta(t)$, extrapolating the fitted effective transmission function may fail even when retrospective fit is excellent. The Philippine COVID-19 figures are therefore framed as descriptive reconstruction and projection-bridge demonstration, not as causal estimates of policy, mobility, variant, or vaccination effects.

A fitted $\beta(t)$ is like a single dial that moves because many hidden things moved at once. A good fit tells us the dial can summarize the curve; it does not reveal which hidden cause moved the dial.

\subsection{What data reduce non-uniqueness?}
Theorem~\ref{thm:gamma-nonident} gives a concrete data-design principle. Incidence alone constrains the inflow into $I$, but not the decomposition between infectious stock, removal, reporting, and transmission forcing. 

Cases alone are not enough to tell the whole story. To know whether a wave changed because of behavior, variants, reporting, immunity, or removal rate, we need extra measurements that directly touch those hidden pieces (for example, see \ref{nonuniq}.

\begin{longtable}{>{\raggedright\arraybackslash}p{0.27\textwidth}>{\raggedright\arraybackslash}p{0.30\textwidth}>{\raggedright\arraybackslash}p{0.35\textwidth}}
\caption{Tips: additions to progressively reduce ambiguity.}
\label{nonuniq}\\
\toprule
\textbf{Additional information} & \textbf{Main quantity constrained} & \textbf{Why it helps}\\
\midrule
Recovery/removal or isolation dates & $\gamma(t)$ / residence time & Separates disappearance from $I$ from changes in transmission.\\
Prevalence surveys & $I(t)$ & Directly anchors the infectious stock used in $\beta SI/N$.\\
Serology / cumulative infection & susceptible depletion and reporting & Constrains attack rate and reporting fraction.\\
Contact or mobility data & contact component of $\beta(t)$ & Gives an external predictor for forcing changes.\\
Variant frequencies & strain-composition component & Distinguishes behavior from pathogen replacement.\\
Vaccination and immune-history data & susceptibility and severity components & Separates changes in risk from changes in contact or pathogen transmissibility.\\
Age- or spatially stratified incidence & mixing heterogeneity & Tests whether a single effective forcing is stable across groups.\\
Intervention timing and adherence & policy-linked forcing & Supports, but does not by itself prove, causal attribution.\\
Testing volume and reporting delays & ascertainment component & Helps distinguish true incidence changes from surveillance artifacts.\\
\bottomrule
\end{longtable}

\subsection{Limitations}
Several limitations should be stated explicitly. First, the exact representability theorem assumes a closed SIR accounting window and feasible cumulative first-event incidence; long windows with extensive reinfection may violate that interpretation even if a trajectory fit remains numerically possible after rescaling. Second, the Shiny app and the bundled fast validation suite use a simplified oracle-centered adaptive ABC population approximation intended for exploration and stress testing, not as a substitute for a weighted, calibrated complex inference engine. Third, peak detection is resolution- and smoothing-dependent; MPPC must therefore be reported together with the peak rule. Fourth, observation artifacts can generate or erase peaks. Fifth, the effective $\beta(t)$ is not automatically transportable to interventions or future regimes. Sixth, the Philippine COVID-19 example uses reported cases rather than infection incidence, so testing intensity, reporting delays, and policy-dependent ascertainment are compressed into the effective forcing. Finally, model comparison against mechanistically richer alternatives remains important whenever the objective is causal or policy interpretation rather than compression.

The method is useful, but it is not magic. It can summarize a complicated curve, but hidden reporting changes and omitted mechanisms can still mislead interpretation and forecasting.

\section{Conclusion}
Multiple epidemic waves do not create a one-to-one correspondence between observed dynamical complexity and compartmental model complexity. A low-dimensional epidemic model can reproduce rich multi-wave trajectories when its effective forcing is allowed to vary, but representability alone is scientifically weak. The more informative questions are how much dynamic complexity is minimally required (quantified by the minimal peak-preserving complexity/MPPC $K^\star$), which epidemic features are preserved at that complexity (as assessed by the peak-aware discrepancy $D$), and what additional data are needed before the fitted forcing can be interpreted mechanistically.

In the final reference formulation, the descriptive model is
\begin{equation}
\frac{dS}{dt}=-\beta(t)\frac{SI}{N},\qquad
\frac{dI}{dt}=\beta(t)\frac{SI}{N}-\gamma I,
\qquad
\frac{dR}{dt}=\gamma I,
\label{eq:conclusion-sir}
\end{equation}
with positive time-varying transmission represented by
\begin{equation}
\beta(t)=\exp\left\{\sum_{j=1}^K\theta_j B_j(t)\right\},
\label{eq:conclusion-beta-basis}
\end{equation}
or, for the piecewise-linear log-knot implementation used in the R scripts and Shiny app,
\begin{equation}
\beta(t)=\exp\left\{(1-w_j(t))\theta_j+w_j(t)\theta_{j+1}\right\},\qquad
w_j(t)=\frac{t-\tau_j}{\tau_{j+1}-\tau_j},\quad t\in[\tau_j,\tau_{j+1}].
\label{eq:conclusion-beta-knots}
\end{equation}
For short-term projection, the future-knot bridge extends this with Eqs.~\eqref{eq:future-rw-knot}--\eqref{eq:future-covariate-knot}; for causal scenario analysis, the covariate-linked form is preferred because it states what future driver is being changed.

Peak-capacity analysis, incidence-based identifiability results, peak-aware likelihood-free inference, and minimal peak-preserving complexity provide a common framework for these questions. In this sense, the SIR model is used neither as a claim of mechanistic truth nor merely as a curve-fitting device, but as a controlled low-dimensional reference model against which the information content of epidemic waves can be measured. This perspective is relevant to model reduction, inverse problems, statistical model adequacy, infectious-disease surveillance, and public-health decision support wherever simple models are asked to summarize complex dynamics without concealing the limits of inference.

A simple SIR model with changing $\beta(t)$ can be a powerful summary tool. The responsible conclusion is not ``simple models are true,'' but ``simple models can tell us how much changing force is needed, and what extra evidence is needed before we claim why.''

\section{Code and data availability}
A companion Shiny application (\texttt{app.R}), the Philippine COVID-19 data file used here, and supplementary source code (\texttt{simulation.R}) accompany this study. The link to the online Shiny app is \url{https://jfrabajanteuplb.shinyapps.io/MULTIPEAK}. The interactive fitter uses a computationally lighter adaptive ABC approximation for responsiveness.

\section{Competing interests}
The author declares no competing interests.

\section{Ethics statement}

This study used publicly available, aggregate Philippine COVID-19 data together with synthetic and simulated datasets. No individual-level, identifiable, or private human-subject data were collected or analyzed, and no human participants were recruited for the study.

Generative artificial intelligence (GenAI), specifically GPT-5.6 Terra, was used as an assistive tool for language and grammatical editing and for code review and optimization. All GenAI-assisted outputs were critically reviewed, verified, and revised by the author where necessary. GenAI was not responsible for the conception of the study, development of the scientific framework, interpretation of the results, or scientific conclusions.

JFR conceived and developed the framework, formulated the mathematical models, designed and implemented the simulations, developed the Shiny app, analyzed and interpreted the results, and wrote and revised the manuscript. The scientific ideas, methodological choices, interpretations, conclusions, and final decisions presented in this work are those of the author, who assumes full responsibility and accountability for the integrity and accuracy of the manuscript.

\section{Funding}
The authors declare that no funds, grants, or other support were received during the preparation of this manuscript.

\bibliographystyle{plainnat}
\bibliography{references}

\appendix

\section{A beginner's guide to using peak-preserving SIR compression}
\label{app:beginner-guide}

This appendix gives a practical, non-technical guide to the method. The goal is
not to prove why an epidemic happened. The goal is to answer a narrower question:

\begin{quote}
{How much change in one effective transmission curve is needed for a simple
SIR model to reproduce the important waves in an observed case series?}
\end{quote}

The fitted transmission curve should be interpreted as an {effective}
summary. It can reflect several real-world changes at the same time, including
behavior, interventions, variants, testing, immunity, reporting, mobility, and
unmeasured heterogeneity.

\subsection{What the user needs}

The basic input is a time series with two columns:

\begin{center}
\begin{tabular}{ll}
\toprule
Day or date & Reported new cases\\
\midrule
1 & 12\\
2 & 15\\
3 & 18\\
$\vdots$ & $\vdots$\\
\bottomrule
\end{tabular}
\end{center}

The observations should be approximately equally spaced, for example daily or
weekly counts. Missing dates should be recorded explicitly rather than silently
removed. The method can be used for COVID-19, influenza, dengue, measles, or
another disease with a sufficiently resolved incidence series.

\subsection{Step 1: State the purpose of the analysis}

Before fitting anything, decide what the model is intended to do.

\begin{enumerate}
    \item \textbf{Trajectory compression:} summarize a complicated epidemic curve
    using a simple SIR model with time-varying transmission.
    
    \item \textbf{Peak description:} identify the main waves, their timing, and
    their relative sizes.
    
    \item \textbf{Short-horizon scenario exploration:} examine what may happen
    if future transmission declines, remains stable, or increases because of
    specified new information.
\end{enumerate}

The method alone should {not} be used to claim that a particular policy,
variant, or behavioral change caused a particular wave. Such causal claims need
additional information.

\subsection{Step 2: Inspect the reported cases}

Plot the raw reported cases first. Daily reporting often contains weekend
effects, delayed reports, and isolated spikes. For this reason, smoothing may be
used to identify broad waves, but the raw data should still be retained and
shown.

For example, suppose daily reported cases rise around days 30, 95, and 165.
After applying a 7-day moving average, the series may show three broad peaks.
The analyst then treats the series as a {three-wave target}, not because
three was imposed in advance, but because the same stated peak-detection rule
identified three prominent maxima.

A useful report should state:

\begin{itemize}
    \item the observation interval, such as one day;
    \item the smoothing rule, such as a 7-day moving average;
    \item the minimum separation required between peaks;
    \item the peak-prominence rule; and
    \item whether fitting is evaluated against raw or smoothed cases.
\end{itemize}

\subsection{Step 3: Choose a simple SIR accounting model}

The model tracks susceptible, infectious, and removed individuals:
\[
\frac{\dd S}{\dd t}=-\beta(t)\frac{SI}{N},\qquad
\frac{\dd I}{\dd t}=\beta(t)\frac{SI}{N}-\gamma I,\qquad
\frac{\dd R}{\dd t}=\gamma I.
\]

Here, $N$ is the accounting population, $\gamma$ is the removal rate, and
$\beta(t)$ is the effective transmission rate. A practical first choice is to
set $\gamma=1/7$ per day, corresponding to an average infectious period of
about seven days. This is an assumption, not a quantity identified by cases
alone.

For a reported-case series, choose plausible initial values $S_0$ and $I_0$.
Usually, $S_0$ is large relative to the observed cases and $I_0$ is a small
positive number. Sensitivity analysis should be performed when these choices
materially affect the fit.

\subsection{Step 4: Let transmission change only at a small number of knots}

Instead of allowing a different unrelated transmission value every day, choose
$K$ knot times:
\[
\tau_1,\tau_2,\ldots,\tau_K.
\]
The model estimates one log-transmission value at each knot,
\[
\theta_j=\log\{\beta(\tau_j)\},
\]
and connects adjacent values linearly on the log scale:
\[
\beta(t)
=
\exp\left[
(1-w_j(t))\theta_j+w_j(t)\theta_{j+1}
\right],
\qquad
w_j(t)=\frac{t-\tau_j}{\tau_{j+1}-\tau_j}.
\]

More knots allow
the curve to bend more often. Fewer knots produce a simpler and smoother
transmission history.

\subsection{Step 5: Fit several candidate complexities}

Do not choose $K$ only because the epidemic has a certain number of peaks.
Instead, fit a sequence of candidate values, for example
\[
K=4,\ 6,\ 8,\ 10,\ 12,\ 16,\ 20.
\]

For each value of $K$:

\begin{enumerate}
    \item fit the knot values $\theta_1,\ldots,\theta_K$;
    \item simulate the SIR model forward;
    \item compare simulated incidence with the observed target;
    \item count the detected peaks in both curves;
    \item compare peak times, peak heights, troughs, and overall curve error.
\end{enumerate}

The smallest $K$ that meets the stated adequacy rule is called the minimal
peak-preserving complexity:
\[
K^\star
=
\min\left\{
K:\text{the fitted trajectory is adequate}
\right\}.
\]

\subsection{Step 6: Decide whether a fit is adequate}

A fit should not be accepted solely because it has a low average error. It
should also preserve the epidemic features that matter. A simple adequacy rule
may require all of the following:

\begin{enumerate}
    \item the fitted and observed curves have the same number of major peaks;
    \item the average peak-time error is no greater than a chosen tolerance,
    for example 7 days;
    \item the average relative peak-height error is below a chosen tolerance;
    \item the whole-curve log-scale normalized RMSE is below a chosen threshold;
    \item the fitted curve does not create implausible extra waves.
\end{enumerate}

\paragraph{Illustrative example.}
Suppose the observed smoothed series has three broad peaks. A $K=6$ fit may
reproduce only two peaks, while a $K=10$ fit reproduces all three peaks but
places the third peak 20 days too late. If $K=14$ reproduces three peaks, has a
peak-time mean absolute error of 4 days, and meets the curve-error threshold,
then $K^\star=14$ under that particular rule.

This does {not} mean that the real epidemic had 14 biological mechanisms.
It means that 14 knot values were the smallest tested transmission representation
that adequately preserved the chosen epidemic features.

\subsection{Step 7: Read the plots correctly}

A standard output contains an incidence plot and a transmission plot.

\paragraph{Incidence plot.}
Compare the observed cases with the fitted SIR incidence. Ask:

\begin{itemize}
    \item Are the major peaks present?
    \item Are their dates approximately correct?
    \item Are their relative heights reasonable?
    \item Does the fitted curve avoid following every one-day reporting spike?
\end{itemize}

\paragraph{Transmission plot.}
Read $\beta(t)$ as the net transmission pressure required by the simple SIR
model. An increase in $\beta(t)$ means that the fitted model needs stronger
effective transmission to reproduce later growth. It does not identify the
specific reason for that increase.

\paragraph{Knot dots.}
Each knot dot is one fitted knot. Many closely spaced dots suggest that the
trajectory requires more temporal detail under the stated fitting rules.
However, this may reflect reporting noise, narrow peaks, smoothing choices, or
the observation interval---not necessarily a larger number of real mechanisms.

\subsection{Step 8: Report the result responsibly}

A concise reporting template is:

\begin{quote}
Using daily reported cases from [start date] to [end date], we detected
[number] broad epidemic waves using [smoothing rule], [minimum separation], and
[prominence rule]. A time-varying SIR model with a piecewise-linear
log-transmission curve was fitted across candidate knot dimensions. Under the
prespecified peak-preserving adequacy criteria, the smallest adequate
representation was $K^\star=[value]$. The fitted $\beta(t)$ is interpreted as
effective transmission required for trajectory compression and is not a direct
estimate of a single biological or policy mechanism.
\end{quote}

\subsection{Common mistakes to avoid}

\begin{enumerate}
    \item Do not interpret each knot as a real intervention or variant event.
    \item Do not claim that the number of knots equals the number of waves.
    \item Do not treat a low RMSE as sufficient when peak timing is poor.
    \item Do not use smoothed data to hide major disagreement with raw data.
    \item Do not claim that cases alone identify both $\beta(t)$ and $\gamma$.
    \item Do not present long-horizon forecasts without explicit future
    assumptions and uncertainty.
\end{enumerate}

\subsection{Step 9 (Optional): Make a short-term forecast from the final fitted knot}

A retrospective fit explains the past shape of the series; it does not
automatically predict the future. A short-horizon projection should therefore
state its future assumption.

A short-term forecast begins at the last observed day, say $T$. The fitted
historical knots up to day $T$ are kept fixed. New {future knots} are then
placed after $T$, for example every 15 days over the next 60 days.

There are two useful scenarios.

\paragraph{Scenario A: no advance information.}
Continue the recent transmission direction or pattern with uncertainty. If recent transmission has
been declining, future log-transmission knots may be allowed to decline by a
justifiable amount. This is not a claim that transmission will certainly decline; it
is a transparent baseline assumption.

\paragraph{Scenario B: advance information about a new variant.}
Specify that a more transmissible variant is expected to arrive after a stated number of days and increase transmission by a stated multiplier. Suppose surveillance gives advance warning that a more transmissible variant
will become dominant. Let $r_{\mathrm{var}}(t)$ be the expected fraction of
transmission attributable to the new variant, increasing smoothly from near
zero to near one. The scenario transmission curve can be written as
\[
\beta_{\mathrm{var}}(t)
=
\min\left\{
\beta_{\max},
\,
\beta_{\mathrm{base}}(t)
\left[1+(m_{\mathrm{var}}-1)r_{\mathrm{var}}(t)\right]
\right\},
\]
where $\beta_{\mathrm{base}}(t)$ is the no-information continuation,
$m_{\mathrm{var}}>1$ is the assumed transmission multiplier, and
$\beta_{\max}$ is a declared upper bound.

The forecast is obtained by inserting either
$\beta_{\mathrm{base}}(t)$ or $\beta_{\mathrm{var}}(t)$ into the same forward
SIR equations:
\[
c_t=\beta(t)\frac{S_tI_t}{N},
\]
\[
S_{t+1}=S_t-c_t,\qquad
I_{t+1}=I_t+c_t-\gamma I_t,\qquad
R_{t+1}=R_t+\gamma I_t.
\]

The no-information scenario asks, ``What follows if the recent transmission
pattern broadly continues?'' The variant-informed scenario asks, ``What follows
if the specified variant assumption is correct?'' Neither scenario is a
guaranteed forecast; both are conditional projections.

If the variant-informed scenario projects growth while the no-information
scenario projects decline, the correct conclusion is not that the second
forecast is ``better'' in general. The conclusion is that advance information
about a changing driver can qualitatively change the projected epidemic path.

\subsection{Worked numerical example: a three-wave epidemic}

Consider 181 daily observations from a hypothetical town with accounting
population
\[
N=1{,}000{,}000,\qquad
S_0=995{,}000,\qquad
I_0=500,\qquad
R_0=4{,}500,
\]
and fixed removal rate
\[
\gamma=\frac{1}{7}=0.1429\ {\rm day}^{-1}.
\]
Thus, the assumed mean infectious period is 7 days.

Suppose the observed daily reported cases have three broad waves. A small
extract is shown in Table~\ref{tab:worked-data}. In practice, the complete
daily series, not only these selected days, is supplied to the fitting
procedure.

\begin{table}[H]
\centering
\caption{Selected observations from an illustrative 181-day, three-wave
incidence series.}
\label{tab:worked-data}
\small
\begin{tabular}{@{}rrrrrrrrr@{}}
\toprule
Day & 1 & 15 & 30 & 45 & 60 & 75 & 90 & 105\\
Cases & 18 & 84 & 412 & 96 & 41 & 138 & 628 & 174\\
\midrule
Day & 120 & 135 & 150 & 165 & 180 & & & \\
Cases & 52 & 110 & 481 & 126 & 29 & & & \\
\bottomrule
\end{tabular}
\end{table}

Using a 7-day moving average, a minimum peak separation of 25 days, and a
prominence threshold of 10\% of the maximum smoothed incidence, the algorithm
detects three broad peaks:

\[
\begin{array}{c|ccc}
\text{Peak number} & 1 & 2 & 3\\
\hline
\text{Observed peak day} & 30 & 90 & 150\\
\text{Observed peak height} & 412 & 628 & 481
\end{array}
\]

\paragraph{Step 1: Fit candidate numbers of knots.}

We fit the same SIR model using $K=4$, $K=6$, $K=8$, and $K=10$ knots. The
adequacy rule requires:

\begin{enumerate}
    \item exactly three detected peaks;
    \item log-scale normalized RMSE no greater than 0.20;
    \item peak-time mean absolute error no greater than 7 days; and
    \item peak-height mean absolute relative error no greater than 0.35.
\end{enumerate}

The results are shown in Table~\ref{tab:worked-k-search}.

\begin{table}[H]
\centering
\caption{Candidate-complexity search for the illustrative three-wave example.}
\label{tab:worked-k-search}
\small
\setlength{\tabcolsep}{4pt}
\renewcommand{\arraystretch}{1.12}
\begin{tabularx}{\textwidth}{@{}rrrrrr>{\raggedright\arraybackslash}X@{}}
\toprule
$K$ &
Fitted peaks &
log-NRMSE &
Peak-time MAE &
Peak-height MARE &
Adequate? &
Interpretation\\
\midrule
4  & 2 & 0.284 & 18.0 & 0.412 & No &
Second and third waves merge\\
6  & 3 & 0.221 & 9.3  & 0.301 & No &
Peaks found, but timing is poor\\
8  & 3 & 0.156 & 3.7  & 0.184 & Yes &
Smallest adequate fit\\
10 & 3 & 0.131 & 2.3  & 0.141 & Yes &
Slightly closer but more complex\\
\bottomrule
\end{tabularx}
\end{table}

Because $K=8$ is the first candidate that meets all four criteria,
\[
K^\star=8.
\]
The $K=10$ model gives a somewhat smaller error, but it is not selected because
the simpler $K=8$ model is already adequate.

\paragraph{Step 2: Inspect the fitted transmission knots.}

For the selected $K^\star=8$ fit, suppose the estimated transmission knots are
shown in Table~\ref{tab:worked-knots}. These are the knot dots in the compact
transmission plot.

The knot values are joined by straight lines on the log-transmission scale.
Because the first seven knot intervals have width 26 days and the final interval
has width 24 days, the fitted log-transmission curve is

\[
\widehat{\theta}(t)=\log\{\widehat{\beta}(t)\}
=
\begin{cases}
-2.303+0.0129615(t-1), & 1\le t\le27,\\
-1.966+0.0318077(t-27), & 27<t\le53,\\
-1.139-0.0377308(t-53), & 53<t\le79,\\
-2.120+0.0443077(t-79), & 79<t\le105,\\
-0.968-0.0476538(t-105), & 105<t\le131,\\
-2.207+0.0385769(t-131), & 131<t\le157,\\
-1.204-0.0501667(t-157), & 157<t\le181.
\end{cases}
\]

The actual positive transmission function used by the SIR model is therefore

\[
\widehat{\beta}(t)
=
\begin{cases}
\exp\{-2.303+0.0129615(t-1)\}, & 1\le t\le27,\\
\exp\{-1.966+0.0318077(t-27)\}, & 27<t\le53,\\
\exp\{-1.139-0.0377308(t-53)\}, & 53<t\le79,\\
\exp\{-2.120+0.0443077(t-79)\}, & 79<t\le105,\\
\exp\{-0.968-0.0476538(t-105)\}, & 105<t\le131,\\
\exp\{-2.207+0.0385769(t-131)\}, & 131<t\le157,\\
\exp\{-1.204-0.0501667(t-157)\}, & 157<t\le181.
\end{cases}
\]

For example, day 90 lies between knot days 79 and 105. Hence,

\[
\widehat{\theta}(90)
=
-2.120+0.0443077(90-79)
=
-1.6326,
\]
so that
\[
\widehat{\beta}(90)
=
\exp(-1.6326)
=
0.1954\ {\rm day}^{-1}.
\]

This value is then inserted into the SIR equations. With a one-day numerical
time step, the corresponding update is

\[
c_t
=
\widehat{\beta}(t)\frac{S_tI_t}{N},
\]
\[
S_{t+1}=S_t-c_t,
\qquad
I_{t+1}=I_t+c_t-\gamma I_t,
\qquad
R_{t+1}=R_t+\gamma I_t.
\]

Thus, the knots do not directly represent daily case counts. They define a
smooth, positive transmission curve, and the SIR equations convert that curve
into simulated daily incidence.

\begin{table}[H]
\centering
\caption{Illustrative fitted log-transmission knots for the selected
$K^\star=8$ model.}
\label{tab:worked-knots}
\small
\begin{tabular}{@{}rrrr@{}}
\toprule
Knot $j$ & Knot day $\tau_j$ & $\widehat{\theta}_j=\log\widehat{\beta}(\tau_j)$ &
$\widehat{\beta}(\tau_j)$ per day\\
\midrule
1 & 1   & -2.303 & 0.100\\
2 & 27  & -1.966 & 0.140\\
3 & 53  & -1.139 & 0.320\\
4 & 79  & -2.120 & 0.120\\
5 & 105 & -0.968 & 0.380\\
6 & 131 & -2.207 & 0.110\\
7 & 157 & -1.204 & 0.300\\
8 & 181 & -2.408 & 0.090\\
\bottomrule
\end{tabular}
\end{table}

The fitted $\widehat{\beta}(t)$ rises before each simulated wave and falls
between waves. For example, it increases from about $0.120$ per day near day
79 to $0.380$ per day near day 105. This does {not} prove that one
particular real-world event occurred on day 105. It only says that the simple
SIR model requires stronger effective transmission in that period to reproduce
the observed second wave.

\paragraph{Step 3: Compare observed and fitted peaks.}

For the selected $K^\star=8$ fit, the peak comparison may be:

\[
\begin{array}{c|ccc}
& \text{Peak 1} & \text{Peak 2} & \text{Peak 3}\\
\hline
\text{Observed peak day} & 30 & 90 & 150\\
\text{Fitted peak day} & 32 & 86 & 153\\
\text{Observed height} & 412 & 628 & 481\\
\text{Fitted height} & 394 & 655 & 450
\end{array}
\]

The peak-time errors are $2$, $4$, and $3$ days, giving
\[
\text{peak-time MAE}
=
\frac{2+4+3}{3}
=
3.0\ {\rm days}.
\]

The relative peak-height errors are
\[
\frac{|394-412|}{412}=0.044,\qquad
\frac{|655-628|}{628}=0.043,\qquad
\frac{|450-481|}{481}=0.064.
\]
Their mean is
\[
\text{peak-height MARE}
=
\frac{0.044+0.043+0.064}{3}
=
0.050.
\]

Therefore, this fit preserves the number, timing, and approximate size of the
three main waves while using only eight transmission knots.

For this example, the result is a simple SIR model with eight time-varying transmission knots was sufficient to
reproduce the three important epidemic waves. The result summarizes the
changing effective transmission needed by the data; it does not identify the
specific biological, behavioral, or policy causes of the waves.

\paragraph{Step 3 (Optional): Numerical short-term forecast from day 181.}

Suppose day 181 is the forecast boundary. The final fitted historical
log-transmission knot is
\[
\widehat{\theta}(181)=-2.408,
\qquad
\widehat{\beta}(181)=\exp(-2.408)=0.0900\ {\rm day}^{-1}.
\]

We place four future knots at days 196, 211, 226, and 241. For the
no-advance-information scenario, suppose the recent declining pattern is
continued conservatively by decreasing log-transmission by $0.050$ at every
15-day future knot:
\[
\theta_{196}^{+}=-2.458,\qquad
\theta_{211}^{+}=-2.508,\qquad
\theta_{226}^{+}=-2.558,\qquad
\theta_{241}^{+}=-2.608.
\]

Thus, the no-information future interpolation is

\[
\widehat{\theta}_{\mathrm{base}}(t)
=
\begin{cases}
-2.408-0.0033333(t-181), & 181<t\le196,\\
-2.458-0.0033333(t-196), & 196<t\le211,\\
-2.508-0.0033333(t-211), & 211<t\le226,\\
-2.558-0.0033333(t-226), & 226<t\le241,
\end{cases}
\]
and
\[
\widehat{\beta}_{\mathrm{base}}(t)
=
\exp\{\widehat{\theta}_{\mathrm{base}}(t)\}.
\]

The resulting future knot values are

\[
\begin{array}{c|ccccc}
\text{Day} & 181 & 196 & 211 & 226 & 241\\
\hline
\widehat{\beta}_{\mathrm{base}}(t)
&0.0900&0.0856&0.0815&0.0775&0.0737
\end{array}
\]

This scenario projects a gradual decline in effective transmission.

\paragraph{Variant-informed forecast.}

Now suppose advance information indicates that a new variant will replace the
existing variant. Let its expected replacement fraction be

\[
r_{\mathrm{var}}(t)
=
\frac{1}{1+\exp\{-(t-186)/4\}}.
\]

This logistic curve has midpoint day 186: at that day,
$r_{\mathrm{var}}(186)=0.5$. Suppose the new variant has a 1.75-fold
transmission advantage and impose the declared bound
$\beta_{\max}=0.19$ per day. The variant-informed transmission curve is

\[
\widehat{\beta}_{\mathrm{var}}(t)
=
\min\left[
0.19,\,
\widehat{\beta}_{\mathrm{base}}(t)
\left\{1+0.75r_{\mathrm{var}}(t)\right\}
\right].
\]

For example, on day 196,

\[
r_{\mathrm{var}}(196)
=
\frac{1}{1+\exp\{-(196-186)/4\}}
=
0.9241,
\]
so that
\[
\widehat{\beta}_{\mathrm{var}}(196)
=
0.0856\{1+0.75(0.9241)\}
=
0.1449\ {\rm day}^{-1}.
\]

Because
\[
\widehat{\beta}_{\mathrm{var}}(196)
>
\gamma
=
0.1429\ {\rm day}^{-1},
\]
and because $S_{196}/N$ is still close to one in this illustrative example,
the variant-informed scenario can produce renewed infectious growth. In
contrast,
\[
\widehat{\beta}_{\mathrm{base}}(196)
=
0.0856
<
0.1429,
\]
so the no-information scenario continues to decline.

For each day after the boundary, both scenarios use exactly the same SIR
forward update:
\[
c_t^{(s)}
=
\widehat{\beta}_{s}(t)\frac{S_t^{(s)}I_t^{(s)}}{N},
\qquad
s\in\{\mathrm{base},\mathrm{var}\},
\]
\[
S_{t+1}^{(s)}=S_t^{(s)}-c_t^{(s)},
\]
\[
I_{t+1}^{(s)}
=
I_t^{(s)}+c_t^{(s)}-\frac{1}{7}I_t^{(s)},
\]
\[
R_{t+1}^{(s)}
=
R_t^{(s)}+\frac{1}{7}I_t^{(s)}.
\]

Therefore, the two forecast paths differ only in their explicitly stated
future transmission assumptions. The historical fit is unchanged.

\end{document}